\documentclass[a4paper,american,draft]{lipics-v2016}
\usepackage{microtype}%
\usepackage{tikz}
\usepackage{framed}
\usepackage{amsmath}
\usepackage{mathabx}
\usetikzlibrary{shapes,automata,patterns,decorations, decorations.pathmorphing}

\title{Symmetric Synthesis\footnote{
This work was partially supported by the Institutional Strategy of the University
of Bremen, funded by the German Excellence Initiative, by the German Research
Foundation (DFG) within the program ``Performance Guarantees for Computer Systems''
and the Transregional Collaborative Research Center ``Automatic Verification and
Analysis of Complex Systems'' (SFB/TR 14 AVACS),
and by the European Research Council (ERC) Grant OSARES (No.~683300).}
}

\author[1]{R\"udiger Ehlers}
\author[2]{Bernd Finkbeiner}
\affil[1]{University of Bremen, Bremen, Germany\\
  \texttt{ruediger.ehlers@uni-bremen.de}}
\affil[2]{Saarland University, Saarbr\"ucken, Germany\\
  \texttt{finkbeiner@cs.uni-saarland.de}}
\authorrunning{R.\,Ehlers and B.\,Finkbeiner} 

\Copyright{R\"udiger Ehlers and Bernd Finkbeiner}


\EventEditors{Satya Lokam and R. Ramanujam}
\EventNoEds{2}
\EventLongTitle{37th IARCS Annual Conference on Foundations of Software Technology and
 Theoretical Computer Science (FSTTCS 2017)}
\EventShortTitle{FSTTCS 2017}
\EventAcronym{FSTTCS}
\EventYear{2017}
\EventDate{December 11–-15, 2017}
\EventLocation{Kanpur, India}
\EventLogo{}
\SeriesVolume{93}
\ArticleNo{26} %

\newcommand\donotshow[1]{}

\newcommand{\NN}{\mathbb{N}}
\newcommand{\ZZ}{\mathbb{Z}}
\newcommand{\AP}{\mathsf{AP}}
\newcommand{\LTLX}{\mathsf{X}}
\newcommand{\LTLF}{\mathsf{F}}
\newcommand{\LTLG}{\mathsf{G}}
\newcommand{\LTLU}{\mathcal{U}}

\newcommand{\newterm}[1]{\emph{#1}}
\newcommand{\TRUE}{\mathbf{true}}
\newcommand{\mylabel}{\label}
\newcommand{\FALSE}{\mathbf{false}}

\DeclareMathOperator{\rot}{rot}
\DeclareMathOperator{\MOD}{mod}

\DeclareMathOperator*{\argmin}{argmin}
\DeclareMathOperator{\rep}{rep}
\DeclareMathOperator{\reps}{rep_\mathrm{S}}

\newcommand{\quotedwedge}{``$\wedge$''}
\newcommand{\quotedvee}{``$\vee$''}

\newtheorem{proposition}[theorem]{Proposition}

\renewcommand\copyrightline\relax
\begin{document}

\maketitle

\begin{abstract}
We study the problem of determining whether a given temporal
specification can be implemented by a symmetric system, i.e., a system
composed from identical components. Symmetry is an important goal in
the design of distributed systems, because systems that are composed
from identical components are easier to build and maintain.  We show
that for the class of rotation-symmetric architectures, i.e.,
multi-process architectures where all processes have access to all
system inputs, but see different rotations of the inputs, the
symmetric synthesis problem is EXPTIME-complete in the number of
processes. In architectures where the processes do not have access to
all input variables, the symmetric synthesis problem becomes
undecidable, even in cases where the standard distributed synthesis problem is decidable.
\end{abstract}

\section{Introduction}

Many classical protocols and distributed systems are
\emph{symmetric}. This means that every process, independently of its
identity, starts in the same initial state and follows the same set of
transitions. Symmetric systems are easier to understand and maintain; especially in VLSI
designs, which usually contain large numbers of identical components,
this is a significant cost factor. Constructing symmetric systems is 
also a step towards building arbitrarily scalable systems \cite{DBLP:conf/podc/EmersonS90,DBLP:journals/toplas/AttieE98,DBLP:journals/corr/JacobsB14}.

There is a large body of results~\cite{DBLP:conf/stoc/Angluin80,Lehmann+Rabin/81/Advantages,Cohen+Lehmann+Pnueli/84/Symmetric,323598,Yamashita+Kameda/88/Computing,Kranakis/96/Symmetry} that deal with the question of which distributed systems need symmetry breaking and which do not. \emph{Leader election} among the
processes on a ring, for example, cannot be implemented symmetrically
\cite{DBLP:conf/stoc/Angluin80}; similarly, in resource-sharing problems, like
the \emph{Dining Philosophers}, the only way to avoid starvation is to
break the symmetry~\cite{Lehmann+Rabin/81/Advantages}.

Our goal is to automate this type of reasoning. Given a specification
of a reactive system in temporal logic, we wish to automatically
determine whether there exists a symmetric implementation.
This is a refinement of the classic
\emph{distributed
  synthesis problem}, which asks whether a temporal specification has an implementation where the processes are arranged in a particular architecture. Distributed synthesis is well-studied~\cite{Wolper/82/Synthesis,DBLP:conf/focs/PnueliR90,Kupferman+Vardi/97/Incomplete,Kupferman+Vardi/00/Mu,Kupferman+Vardi/01/Synthesizing,DBLP:conf/lics/FinkbeinerS05,DBLP:conf/atva/ScheweF07a}. However, the
approach presented in this paper is the first to 
synthesize \emph{symmetric} implementations.
  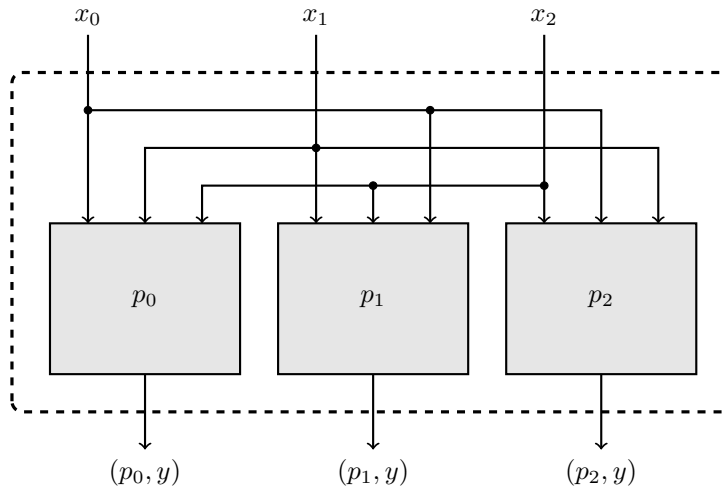
\begin{figure}
   \centering\begin{tikzpicture} 
    \draw[very thick, dashed, rounded corners] (0,0) rectangle (9.5,4.5);
    \draw[thick,fill=black!10!white] (0.5,0.5) rectangle +(2.5,2); \draw[thick,fill=black!10!white] (3.5,0.5) rectangle +(2.5,2); \draw[thick,fill=black!10!white] (6.5,0.5) rectangle +(2.5,2); \node at (1.75,1.5)
    {$p_0$}; \node at (4.75,1.5)
    {$p_1$}; \node at (7.75,1.5)
    {$p_2$};
    \draw[thick,->] (1,5.0) -- node[at start,above]
    {$x_0$} (1,2.5); \draw[thick,->] (4,5.0) -- node[at start,above]
    {$x_1$} (4,2.5); \draw[thick,->] (7,5.0) -- node[at start,above]
    {$x_2$} (7,2.5);
    \draw[thick,->] (7,3) -- (2.5,3) -- (2.5,2.5); \draw[thick,<->] (1.75,2.5) -- (1.75,3.5) -- (8.5,3.5) -- (8.5,2.5); \draw[thick,->] (1,4) -- (7.75,4) -- (7.75,2.5); \draw[thick,->] (5.5,4) -- (5.5,2.5); \draw[thick,->] (4.75,3) -- (4.75,2.5); \draw[fill=black] (4,3.5) circle (0.5mm); \draw[fill=black] (5.5,4) circle (0.5mm); \draw[fill=black] (4.75,3) circle (0.5mm); \draw[fill=black] (7,3) circle (0.5mm); \draw[fill=black] (1,4) circle (0.5mm); 
    \draw[thick,->] (1.75,0.5) -- node[at end,below]
    {$(p_0,y)$} (1.75,-0.5); \draw[thick,->] (4.75,0.5) -- node[at end,below]
    {$(p_1,y)$} (4.75,-0.5); \draw[thick,->] (7.75,0.5) -- node[at end,below]
    {$(p_2,y)$} (7.75,-0.5);
   \end{tikzpicture}
   \caption{A simple rotation-symmetric architecture.}
   \label{fig:simpleRotationSymmetricArchitecture}
  \end{figure}
  We consider \emph{rotation-symmetric} system
  architectures. Rotation-symmetric architectures are multi-process
  architectures where all processes have access to all system inputs,
  but see different rotations of the inputs.
  Figure~\ref{fig:simpleRotationSymmetricArchitecture} shows a simple
  rotation-symmetric architecture. Rotation-symmetric architectures
  are suitable to reason about distributed systems that lack a central
  coordination process.  They can, for example, model leader election
  scenarios and distributed traffic light
  controllers~\cite{DBLP:books/daglib/0033466}. The fact that the
  processes obtain their input in different rotations is important:
  since all processes have the same
  implementation, they would otherwise also produce the same output.  The
  synthesis problem for such systems could trivially be reduced to the
  standard synthesis problem by adding a constraint that the outputs
  are the same all the time.
  
  We present an algorithm for the synthesis of symmetric systems in rotation-symmetric architectures from specifications in linear-time temporal logic (LTL).
  Most standard synthesis algorithms follow the automata-theoretic
approach~\cite{Rabin1972}, whereby the given temporal formula is translated into a tree
automaton that accepts exactly those computation trees that satisfy
the formula.  Hence, the specification is realizable if and only if
the language of the automaton is non-empty. The synthesis algorithm then simply
extracts some finite-state implementation
from the language of the automaton.  The situation is more
difficult when we wish to decide the existence of a symmetric
solution, because the language of the automaton may contain both
computation trees that belong to symmetric implementations and
computation trees that belong to asymmetric implementations.  As we
show in Section~\ref{sec:rotationSymmetricSynthesis}, symmetry is \emph{not} a regular property: we 
therefore cannot check symmetry with a separate tree automaton or
encode symmetry as a temporal logic formula and add it to the
specification.

The key insight of our algorithm is that the
paths in the computation trees produced by symmetric implementations
are guaranteed to be invariant under rotations: if, in each position
of two (finite or infinite) computation paths, the values of the
input variables of the $j$th process in the first path correspond to the values of
the input variables of the $((j+k) \MOD n)$th process, for some $k$, in the second
path, then the values of the output variables of the $j$th process must also, in each
position, correspond to the values of the output variables of the
$({(j+k) \MOD n})$th process (for all $0 \leq j < n$, where $n$ is the number of processes).
Our algorithm exploits this observation to simplify the computation
trees. Paths that are just rotations of each other are collapsed into
a single representative. Computations in different processes that must
lead to identical outputs are thus kept in the same path of the
reduced tree; the paths only split when the symmetry is broken by some
input. While symmetry is difficult to check on the original
computation tree, it becomes a local condition on individual paths in
the reduced tree: as long as the output never spontaneously introduces
asymmetry, i.e., as long as every asymmetry in the output can be
explained by a previous asymmetry in the input, the reduced tree can
be expanded into a full computation tree that we know, by
construction, to be symmetric.

As we show in Section~\ref{sec:synthesis},
the running time of our synthesis algorithm is single-exponential in the number of processes. In Section \ref{sec:complexity}, we show that our algorithm is asymptotically optimal: the problem is EXPTIME-complete in the number of processes.
In Section~\ref{sec:undecidability}, we study the extension of the synthesis problem to
the case where the processes no
longer have access to all variables. Here, our result is
negative: under incomplete information, the symmetric synthesis
problem is undecidable even for system architectures where the
standard synthesis problem is decidable.
This paper is based on previously unpublished results from the first author's PhD thesis~\cite{DBLP:books/daglib/0033466}, where also additional details of the presented results can be found.

\section{Preliminaries}

A \newterm{reactive system} produces a valuation to the output propositions in some set $\AP^O$ and reads the values of the input propositions in some set $\AP^I$ in every step of its execution. The behavior of a reactive system can be described as a \newterm{computation tree} $\langle T, \tau \rangle$, where $T = (2^{\AP^I})^*$ is the set of tree nodes and $\tau : T \rightarrow 2^{\AP^O}$ labels every tree node $t$ by the output propositions $\tau(t)$ that the system sets to $\TRUE$ after having read $t$ as its (prefix) input sequence.

A \newterm{trace} in a computation tree $\langle T, \tau \rangle$ is an infinite sequence $(\tau(\epsilon) \cup t_0) (\tau(t_0) \cup t_1) (\tau(t_0 t_1) \cup t_2) (\tau(t_0 t_1 t_2) \cup t_3) \ldots \in (2^{\AP^I \cup \AP^O})^\omega$. 
Given some language $L \subseteq (2^{\AP^I \cup \AP^O})^\omega$, \newterm{reactive synthesis} is the process of checking if there exists a computation tree $\langle T, \tau \rangle$ with $T = (2^{\AP^I})^*$ as node set such that every trace of $\langle T, \tau \rangle$ is in $L$.
A classical logic to denote specification languages is linear temporal logic (LTL, \cite{DBLP:conf/focs/Pnueli77}). LTL formulas for reactive system specifications are built according to the grammar
\begin{equation*}
\varphi :== p \mid \neg \varphi \mid \varphi \vee \varphi \mid \varphi \wedge \varphi \mid \LTLX \, \varphi \mid \LTLG\,\varphi \mid \LTLF \, \varphi \mid \varphi \, \LTLU \, \varphi,
\end{equation*}
using the temporal operators $\LTLG$ (\newterm{globally}), $\LTLF$ (\newterm{eventually}), $\LTLX$ (\newterm{next}), and $\LTLU$ (\newterm{until}). All elements from $\AP^I$ and $\AP^O$ can be used as propositions $p$. A more formal definition of LTL is given in \cite{DBLP:conf/focs/Pnueli77,Clarke1999}.

For LTL specifications, it is known that if and only if there exists a computation tree all of whose traces satisfy a specification (i.e., the specification is \newterm{realizable}), there exists a \newterm{regular} such computation tree. A computation tree is regular if it has only finitely many different sub-trees. 
Given a computation tree $\langle T, \tau \rangle$, a tree $\langle T', \tau' \rangle$ is a \newterm{sub-tree} of $\langle T, \tau \rangle$ if and only if $T = T'$ and there exists a $\hat t \in T$ such that for every $t \in T$, we have $\tau'(t) = \tau(\hat t t)$.
Regular computation trees can be translated to \newterm{finite-state machines} and implemented in hardware or software using a finite amount of memory.
A tree language for some sets $\AP^I$ and $\AP^O$ is a subset of all trees $\langle T, \tau \rangle$ with $T = ({2^{\AP^I}})^*$ and $\tau : T \rightarrow 2^{\AP^O}$. A tree or word language is called \newterm{regular} if it can be recognized by some finite tree or word automaton (with a \newterm{Muller acceptance condition}, see \cite{DBLP:conf/birthday/2008thomas} for details). 

In \newterm{distributed synthesis}, we search for a distributed implementation of a finite state-machine. Given is an architecture that defines several \newterm{processes} and the \newterm{signals} that connect the processes among themselves and with the global input and output of the architecture. Starting from a specification over all signals, we search for implementations for all of the processes such that the computation tree \newterm{induced} by the process implementations and the architecture satisfies the specification. In the induced computation tree, all processes are executed at the same time and in parallel, using the usual parallel composition semantics.

It is known since the seminal work by Pneuli and Rosner \cite{DBLP:conf/focs/PnueliR90} that not all architectures have a decidable distributed synthesis problem. Figure~\ref{fig:A0Architecture} depicts the \newterm{A0 architecture} that they defined as an example for an undecidable architecture. 
Finkbeiner and Schewe~\cite{DBLP:conf/lics/FinkbeinerS05} later proved that the distributed synthesis problem is decidable if and only if there exists no \newterm{information fork} in the architecture.
An information fork is a pair of processes that are incomparably informed, i.e., for which each of the processes has access to some global input that the other process cannot read.
For a more formal definition of distributed synthesis, the interested reader is referred to \cite{DBLP:conf/lics/FinkbeinerS05}.

A \newterm{Turing machine} is a tuple $M=(Q,\Sigma,\Gamma, \allowbreak \delta, \allowbreak q_0, \allowbreak g)$ in which $Q$ is a finite set of states, $\Sigma$ is an input alphabet, $\Gamma \supseteq \Sigma$ is a (finite) tape alphabet, $\delta : Q \times \Gamma \rightarrow (Q \times \Gamma \times \{-1,0,1\})^2$ encodes the Turing machine transition function, $q_0 \in Q$ is an initial state, and $g$ maps every state to its type, which can be \newterm{accepting}, \newterm{rejecting}, or \newterm{transient}. The $\delta$ function maps every state/tape content combination to exactly two possible successor state/tape content/tape motion combinations. For deterministic Turing machines, the two successor combinations are always the same.
\newterm{Alternating Turing machines} \cite{DBLP:journals/jacm/ChandraKS81} extend the non-deterministic Turing machines by  partitioning the transient states into  \newterm{universally branching} and \newterm{existentially branching states}. An (alternating) Turing machine accepts a word $w \in \Sigma^*$ if there exists an accepting run tree when starting in state $q_0$ with the tape empty except for a copy of $w$ where the machine head starts on the first character of $w$. In all universal states, the Turing machine execution must be accepting for both possible transitions.

We assume that the \newterm{modulo function} always returns a non-negative number, such that, e.g., $-13 \mod 5 = 2$.

\section{The Symmetric Synthesis Problem}

We consider distributed reactive synthesis problems in which all processes share the same implementation. A process has an \newterm{interface} $\mathcal{N} = (\AP^I,\AP^O)$ with the local input proposition set $\AP^I$ and a local output proposition set $\AP^O$. The connections between the processes are described in an \newterm{architecture}.
  \begin{definition}[Symmetric architecture]
   Given an interface
   $\mathcal{N} = (\AP^I,\AP^O)$, a \newterm{symmetric architecture} over
   $\mathcal{N}$ is a tuple
   $\mathcal{E} =
   (S,P,\AP^I_G,E^\mathit{in},E^\mathit{out})$ with:
   \begin{itemize}
   \item
    the set of (internal) signals
    $S$,
   \item
    the \newterm{process set}
    $P$, 
   \item
    the \newterm{global input signal set}
    $\AP_G^I$, 
   \item
    the \newterm{input edge function}
    $E^\mathit{in} : (P \times \AP^I) \rightarrow (S \cup \AP^I_G)$, and
   \item
    the \newterm{output edge function}
    $E^\mathit{out} : (P \times \AP^O) \rightarrow
    S$.
   \end{itemize}
  \end{definition}
As an example, the architecture given in the right part of Figure~\ref{fig:S0Architecture} hosts processes with the interface $\mathcal{N} = (\{a\},\{b\})$ and has the components $S = \{y,z\}$, $P = \{0,1\}$, $\AP^I_G = \{x\}$, $E^\mathit{in} = \{ (0,a) \mapsto x, (1,a) \mapsto y\}$, and $E^\mathit{out} = \{ (0,b) \mapsto y, (1,b) \mapsto z\}$.
We only consider architectures in which every internal signal is written to by exactly one local output of one process. Given a FSM for a process with an interface $\mathcal{N}$ and an architecture $\mathcal{E} =
   (S,P,\AP^I_G,E^\mathit{in},E^\mathit{out})$ over $\mathcal{N}$, we can construct an FSM with $\AP^I_G$ as input proposition set and $S$ as output proposition set that implements the behavior of the complete architecture when using the FSM as process implementation. Without loss of generality, we use the standard synchronous composition semantics to do so.
We define the symmetric synthesis problem as follows:
\begin{definition}
\label{def:symmetricSynthesisProblem}
Given an interface $\mathcal{N} = (\AP^I,\AP^O)$, an architecture $\mathcal{E} =
   (S,P,\AP^I_G,E^\mathit{in}, \allowbreak{} E^\mathit{out})$, and a specification $\varphi$ over the propositions $\AP^I_G \cup S$, 
    \newterm{the symmetric synthesis problem} is to check if an FSM implementation $\mathcal{F}$ with the input proposition set $\AP^I$ and output proposition set $\AP^O$ exists such that the FSM obtained by plugging $\mathcal{F}$ into $\mathcal{E}$ satisfies $\varphi$. In case of a positive answer, we also want to obtain $\mathcal{F}$.
\end{definition}   

\section{Rotation-Symmetric Synthesis}
\label{sec:rotationSymmetricSynthesis}
\label{sec:synthesis}

Many symmetric architectures found in practice consist of a ring of processes, all of which read all the input to the overall system. A slight generalization of this architecture shape is the class of \newterm{rotation-symmetric architectures}.

  \begin{definition}
   \label{def:rotationSymmetricArchitecture} A symmetric architecture
   $\mathcal{E} =
   (S,P,\AP^I_G,E^\mathit{in},E^\mathit{out})$ over the interface
   $\mathcal{N} = (\AP^I,
   \AP^O)$ with $n$ processes is called \newterm{rotation-symmetric} if and only if there exists a \newterm{local designated proposition set} $\AP^I_L$ for every process instance such that the following conditions hold:
   \begin{itemize}
   \item
    $\AP^I_G = \AP^I = \AP^I_L \times \{0, \ldots,
     {n-1}\}$ and $P = \{p_0, \ldots, p_{n-1}\}$.
   \item
    $S = \AP^O \times \{0, \ldots, n-1\}$
   \item
    for every
    $p_i \in P$, every $x \in \AP^I_L$, and
    every $j \in \{0, \ldots,
     n-1\}$, we have
    $E^\mathit{in}(p_i,(x,j)) = (x,(j-i) \mod n)$, and
   \item
    for every
    $x \in
    \AP^O$ and $p_i \in P$, we have
    $E^\mathit{out}(p_i,x) =
    (x,i)$.
   \end{itemize}
  \end{definition}
We show in this section that the symmetric synthesis problem for rotation-symmetric architectures and linear-time temporal logic (LTL) is decidable.

The key observation that we use to prove decidability is that the computation trees that characterize the input/output behavior of a process implementation plugged into a rotation-symmetric architecture have a useful property that we call the \newterm{symmetry} property. While this property is non-regular and thus cannot be encoded into the specification (Lemma~\ref{lem:nonRegular}), we show how to decompose it into two sub-properties, one of which is regular. The other one is still non-regular, but has the advantage that we can enforce it in a synthesis process by post-processing the computation tree obtained from a synthesis procedure to contain only rotations of the computation tree paths along so-called \newterm{normalized inputs}. Since every tree with the symmetry property is left unaltered by this step and we also describe how to ensure that the result of the post-processing step is guaranteed to be a correct solution, this approach is sound and complete.

We assume some fixed rotation-symmetric architecture $\mathcal{E} = (S,P,\AP^I_G,E^\mathit{in}, \allowbreak E^\mathit{out})$ over some local process interface $(\AP^I_L,\AP^O)$ to be given, define $\mathcal{I} = 2^{\\AP^I_G}$ to denote the global input alphabet to all processes, while $\mathcal{O} = 2^{\{\AP^O \times \{0, \ldots, n-1\}\}}$ \label{def:mathcalO} denotes the global output.
The local output of one process is given as $O = 2^{\AP^O}$.

The following rotation function will become useful in the analysis below. Let
  $U = 2^{\AP \times
   \{0, \ldots, n-1\}}$ for some other set
  $\AP$.
  We define a \newterm{rotation operator} 
  $\rot : U \times \ZZ \rightarrow
  U$ with
  $\rot(u,k) = \{(p,(j+k) \MOD n) \mid (p,j) \in
   u\}$ for every
  $u \in
  U$ and
  $k \in
  \ZZ$.
  Furthermore, we extend the
  $\mathrm{rot}$ function to LTL formulas and define
  $\mathrm{rot}(\psi,k)$ for an LTL formula
  $\psi$ over the set of propositions
  $\AP \times
  \{0, \ldots, n-1\}$ and
  $k \in
  \ZZ$ to be
  $\psi$ with all atomic propositions
  $(p,j)$ replaced by
  $(p,(j+k) \MOD
  n)$ for
  $p \in
  \AP$,
  $j \in
  \ZZ$.
  For cla\-ri\-ty, when dealing with the
  $\rot$ function for some set
  $U = 2^{\AP \times
   \{0, \ldots, n-1\}}$, we often partition the elements of
  $\AP \times
  \{0, \ldots, n-1\}$ by their process indices and for example write
  $(X_0, \ldots, X_{n-1})$ instead of
  $(X_0 \times \{0\}) \cup \ldots \cup
  (X_{n-1} \times \{ n-1 \})$ for $X_0, \ldots, X_{n-1} \subseteq \AP$.
  The rotation function is extended to sequences of elements in
  $U$ by rotating the individual sequence items.

  \begin{definition}[Symmetry property]
   Given a tree
   $\langle T, \tau
   \rangle$ over
   $T =
   \mathcal{I}^*$ and
   $\tau : T \rightarrow
   \mathcal{O}$, we say that the tree has the \newterm{symmetry property} if for each
   $t \in
   T$ and
   $0 \leq i <
   n$,
   $\tau(\rot(t,i)) =
   \rot(\tau(t),i)$.
  \end{definition}
 \begin{lemma}[Symmetry lemma]
   \label{SymmetryLemma} The set of regular trees having the symmetry property is precisely the same as the set of trees that are induced by a rotation-symmetric architecture for some process implementation.
\end{lemma}
A proof of the lemma can be found in the appendix. The symmetry property is not a regular tree property, and hence cannot be encoded into a tree or word automaton.

  \begin{lemma}
  \label{lem:nonRegular}
   The set of symmetric computation trees for the two-process rotation-symmetric architecture
   with process interface $\mathcal{N}=(\AP^I_L \times \{0,1\}, \AP^O)$ and $\AP^I_L = \{i\}$ and $\AP^O = \{o\}$ is not a regular tree language. 
  \end{lemma}

  \begin{proof}
   For a proof by contradiction, suppose that the set of symmetric computation trees is regular.
   The language includes a tree with the symmetry property in which the node labels on the path
   $(\emptyset,\{i\})^*$ and, symmetrically, on the path
   $(
   \{i\},\emptyset)^*$ form the sequence
   $l=( \emptyset, \emptyset )^1 ( \{o\},\{o\})  ( \emptyset, \emptyset )^2 ( \{o\},\{o\})
   \ldots$, i.e., the length of the
   $( \emptyset, \emptyset )$-sequences grows according to the distance to the root.
   According to the pumping lemma for regular tree languages, however, the sequence
   $l$ can be partitioned into
   $l=u \cdot v \cdot
   w$, such that, for every
   $k>0$, there exists a tree in the language where the label sequence on
   $(
   \emptyset,\{i\})^*$ is
   $l=u \cdot v^k \cdot
   w$, while the label sequence on
   $(
   \{i\},\emptyset)^*$ is
   still~$l$.
   Clearly, these trees are not symmetric.
  \end{proof}
Since the symmetry property is non-regular, we need to alter the synthesis process itself to account for it. In order to synthesize an implementation for \emph{one} process, we synthesize implementations for \emph{all} processes together. 
These only need to work correctly on \newterm{normalized input sequences} $t \in \mathcal{I}^*$. An input sequence is \newterm{normalized} if $\min_i \rot(t,i) = t$, where the $\min$ function uses the lexicographic ordering over the strings in $\mathcal{I}^*$. For the ordering of the elements in $\mathcal{I}$, we consider the lexicographic ordering of their tuple representation. For example, we have $(0,1,0) < (0,1,1)$ and $(0,1,0) < (1,0,0)$ for a three-process architecture.
A tree with the symmetry property is fully determined by the labels along normalized input sequences, as for every non-normalized input sequence $t' \in \mathcal{I}^*$, we have $\tau(t') = \rot(\tau(t),i)$ for every $i$ such that $t' = \rot(t,i)$.

When only considering the normalized input sequences during synthesis, we can take the computation tree for all processes in the architecture together and complete it by filling all other tree labels with rotations of the tree labels along normalized inputs. We call the resulting tree its \newterm{symmetric completion}\label{text:ImplicitDefSymmetricCompletion}. If afterwards, we have $\tau(\rot(t,i)) =
   \rot(\tau(t),i)$ for all $t \in \mathcal{I}^*$ and $i \in \NN$, then the symmetry lemma guarantees that the resulting tree is induced by some process instantiated in a rotation-symmetric architecture. So if we can guarantee that (1) $\tau(\rot(t,i)) = \rot(\tau(t),i)$ is actually the case for all normalized $t$ and $i \in \NN$ and (2) that the symmetric completion of the tree satisfies the specification along all paths, then we can obtain a correct process implementation by synthesizing a computation tree for the complete architecture. Our construction for symmetric synthesis consist of these two components, which we describe in more detail below.

\subsection{Ensuring Symmetric Completability}
\label{subsec:definitionOfReps}
Not every $\mathcal{O}$-labeled computation tree can easily be made symmetric by replacing the tree labels for non-normalized input sequences. Take for example a tree $\langle T, \tau \rangle$ for the architecture given in Figure~\ref{fig:simpleRotationSymmetricArchitecture} with $\tau(\epsilon)=(\emptyset,\emptyset,\{y\})$. Since the output of the processes is initially different, this means that they cannot have the same implementation. We show in this section that detecting such cases is simple, and the formalization of the observation is a regular property that can be easily encoded into LTL. 
\begin{definition}
   \label{def:reAndReps} Let
   $\AP$ be some set, and
   $P = \{p_0, \ldots,
    p_{n-1}\}$ be a list of process identifiers.
   For every
   $x \subseteq (\AP \times
   \{0, \ldots, n-1\})$ and
   $w = w_0 w_1 w_2 \ldots w_l \in (2^{\AP \times
    \{0, \ldots, n-1\}})^*$, we define
   \begin{align*}
    \rep(x) = &
    \ |\{j \in \{0,\ldots,n-1\} \mid \rot(x,j) = x \} |\\
    \reps( \epsilon ) = &
    \ n \\
    \reps(w) = &
    \,\gcd(\reps(w_0 \ldots w_{n-1}),\rep(w_n)),
   \end{align*}
   where $\gcd$ denotes the \newterm{greatest common divisor} function.
  \end{definition}
For some word $t \in \mathcal{I}^*$, $\reps(t)$ represents how many different rotations in $\{0,\ldots, n-1\}$ of $t$ exist that map the word to itself.

\begin{lemma}[Second symmetry lemma]
   \label{lem:symmetryLemma2} Let
   $\langle T, \tau
   \rangle$ be a computation tree with
   $T =
   \mathcal{I}^*$ and
   $\tau : T \rightarrow
   \mathcal{O}$ for which for every
   $t \in
   T$, we have that
   $\reps(t) \divides \reps(\tau(t))$ (where the $\divides$ symbol refers to division without remainder).
   The unique symmetric completion of
   $\langle T, \tau
   \rangle$ has the symmetry property.
   Furthermore, if
   $\langle T, \tau
   \rangle$ is regular, then so is
   its unique symmetric completion.
  \end{lemma}
By the second symmetry lemma, it suffices for a computation tree to have $\reps(t) \divides \reps(\tau(t))$ for all $t \in T$ to ensure that the symmetric completion of the tree has the symmetry property. We can encode this requirement in LTL as
\begin{align*}
\varphi_{\mathit{outcond}} = \bigwedge_{d \in \{1, \ldots, n\}, d \divides n} \neg (\mathrm{sym}(\mathcal{I},d,n) \, \LTLU \, \neg \mathrm{sym}(\mathcal{O},d,n))
\end{align*}
for the function
\begin{align*}
\mathrm{sym}(\AP,d,n) = \bigwedge_{a \in \AP, j \in \{0, \ldots, n-1\}} (a,j) \leftrightarrow (a,j+\frac{n}{d})
\end{align*}
that encodes, for each $i \subseteq \AP \times \{0, \ldots, {n-1}\}$ whether $d \mid \rep(i)$ (for $d \in \NN$ with $d \divides n$).

\subsection{Ensuring That the Tree Completion Satisfies the Specification}
\label{subsec:ensuringThatTheTreeCompletionSatisfiesTheSpec}
If we have a computation tree $\langle T, \tau \rangle$ all of whose traces satisfy some linear-time specification $\varphi$, this does not imply that its rotation-symmetric completion satisfies $\varphi$ as well. If all traces of $\langle T, \tau \rangle$ however satisfy $\varphi \wedge \rot(\varphi,1) \wedge \ldots, \wedge \rot(\varphi,n-1)$, then since we know that every infinite trace in the rotation-symmetric completion is a rotation of a trace in the original tree by some value $i \in \NN$, we know that the rotation-symmetric completion also satisfies $\varphi$ along every trace. 
So if we synthesize a tree for $\varphi' = \varphi \wedge \rot(\varphi,1) \wedge \ldots \wedge \rot(\varphi,n-1)$ as specification instead of $\varphi$, taking the rotation-symmetric completion maintains $\varphi$.

Note that strengthening $\varphi$ to $\varphi'$ comes without loss of generality if we are interested in rotation-symmetric implementations. By the symmetry property, if the tree $\langle T, \tau \rangle$ induced by a rotation-symmetric architecture and a process implementation satisfies $\varphi$, then it also satisfies $\rot(\varphi,i)$ for all $i \in \NN$ as every rotation of every trace in the tree is also a trace in the tree. Hence, to satisfy $\varphi$, it also needs to satisfy $\rot(\varphi,i)$ as otherwise we could take a trace not satisfying $\rot(\varphi,i)$, rotate it by $-i$, and obtain a trace that does not satisfy $\varphi$.

\subsection{Putting Everything Together}

Using the concepts defined above, we are now ready to tie them together to a complete synthesis process. We start with a specification $\varphi$ over the architecture input propositions $\AP^I_G$ and the output proposition set $\AP^O \times \{0, \ldots, n-1\}$ for $|P|=n$.
\begin{enumerate}
\item We modify the specification $\varphi$ to $\varphi' = \varphi \wedge \rot(\varphi,1) \wedge \ldots \wedge \rot(\varphi,n-1)$.
\item We modify $\varphi'$ to $\varphi'' = \varphi' \wedge \varphi_{\mathit{outcond}} $ (as described in Section~\ref{subsec:ensuringThatTheTreeCompletionSatisfiesTheSpec}).
\item We synthesize a regular tree $\langle T, \tau \rangle$ that satisfies $\varphi''$ along all paths using a classical reactive synthesis procedure. If there is no such tree, the specification is unrealizable.
\item If a regular computation tree $\langle T, \tau \rangle$ is found, we replace every label along non-normalized directions by rotations of $\tau$'s labels along normalized directions to get a tree $\langle T', \tau' \rangle$ with the symmetry property.
\item We cut off the labels of $\tau'$ except for the output of the first process in the architecture. The resulting (regular) tree is the synthesized process implementation.
\end{enumerate}

\begin{proposition}
The above synthesis process from LTL has a complexity that is 2EXPTIME in the length of the specification and exponential-time in the number of processes.
\end{proposition}
\begin{proof}
We use the automata-theoretic approach to reactive system synthesis from \cite{DBLP:conf/focs/KupfermanV05,DBLP:conf/atva/ScheweF07a} and the concepts defined in these works.
We start by translating the specification to a universal co-B\"uchi word (UCW) automaton, which is of size $2^{O(|\varphi|)}$ in the size of the specification. 
As UCWs do not blow up under conjunction, executing step 1 from the construction above leads to an automaton of size $n \cdot 2^{O(|\varphi|)}$. A deterministic automaton for the added property in step 2 can be built with at most $n$ states, so executing step 2 leads to at most $n$ additional states, and we obtain an automaton with $n+n \cdot 2^{O(|\varphi|)} = n \cdot 2^{O(|\varphi|)}$ many states. The \emph{bounded synthesis} approach works with specifications given as co-Büchi word automata \cite{DBLP:conf/atva/ScheweF07a} and takes time exponential in the number of states of the automaton. The overall time complexity so far is thus 2EXPTIME in $|\varphi|$ and exponential in $n$. Step 4 leads to a blow-up of at most a factor of $n^2$ and can be done in time polynomial in the number of states in the synthesized finite-state machine (whose size is proportional to the time complexity of the synthesis procedure executed in the previous step). Step 5 is simple and takes time linear in the size of the FSM.
\end{proof}
Note that even though the construction above discards all non-normalized parts of the synthesized computation tree, asking the synthesis algorithm to nevertheless synthesize these parts according to the specification comes without loss of generality, as trees with the symmetry property (which we are actually searching for) fulfill $\varphi''$ along all paths if all of their paths satisfy $\varphi$. So the synthesis process does not report spurious unrealizability.

\section{Rotation-Symmetric Synthesis -- Complexity}
\label{sec:complexity}

  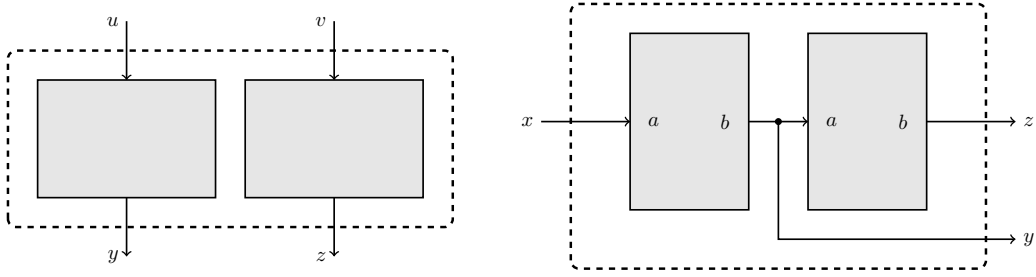
\begin{figure}
    \centering
\scalebox{0.78}{
    \begin{tikzpicture}
    \draw[very thick, dashed, rounded corners] (0,0) rectangle (7.5,3);
    \draw[thick,fill=black!10!white] (0.5,0.5) rectangle +(3,2); \draw[thick,fill=black!10!white] (4,0.5) rectangle +(3,2);
    \draw[thick,->] (2,3.5) -- node[at start,left] {$u$} (2,2.5); \draw[thick,->] (5.5,3.5) -- node[at start,left] {$v$} (5.5,2.5);
    \draw[thick,->] (2,0.5) -- node[at end,left] {$y$} (2,-0.5); \draw[thick,->] (5.5,0.5) -- node[at end,left] {$z$} (5.5,-0.5);
   \end{tikzpicture}}\qquad\scalebox{0.78}{
    \begin{tikzpicture}
    \draw[very thick, dashed, rounded corners] (0,0) rectangle (7,4.5);
    \draw[thick,fill=black!10!white] (1,1) rectangle +(2,3); \draw[thick,fill=black!10!white] (4,1) rectangle +(2,3);
    \draw[thick,->] (-0.5,2.5) -- node[at start,left] {$x$} (1,2.5);
    \draw[thick,->] (3,2.5) -- (4,2.5); \draw[fill] (3.5,2.5) circle (0.05cm); \draw[thick,->] (3.5,2.5) -- (3.5,0.5) -- node[at end,right] {$y$} (7.5,0.5);
    \draw[thick,->] (6,2.5) -- node[at end,right] {$z$} (7.5,2.5);
    \node at (1.4,2.5) {$a$}; \node at (2.6,2.5) {$b$}; \node at (4.4,2.5) {$a$}; \node at (5.6,2.5) {$b$};
   \end{tikzpicture}}
   \caption{System architectures with undecidable synthesis problems. On the left: architecture A0, as defined by Pnueli and Rosner~\cite{DBLP:conf/focs/PnueliR90}; on the right: the symmetric architecture S0. The distributed synthesis problem of A0 and the symmetric synthesis problem of S0 are undecidable.}
   \label{fig:S0Architecture}
   \label{fig:A0Architecture}
  \end{figure}

The symmetric synthesis construction from the previous section has a time complexity that is doubly-exponential in the length of the specification and singly-exponential in the number of processes. We want to show in this section that this matches the complexity of the problem by giving a corresponding hardness result.
The 2EXPTIME-hardness in the specification length is inherited from the complexity of LTL synthesis \cite{DBLP:conf/icalp/PnueliR89}. For the EXPTIME complexity in the number of processes, we provide the following result:

  \begin{lemma}
   \label{lem:symmetricSynthesisHardness}
   Given an
   $f(k)$-space bounded alternating Turing machine $M=(Q,\Sigma,\Gamma, \allowbreak \delta, \allowbreak q_0, \allowbreak g)$, 
   we can reduce the acceptance of a word
   $w \in
   \Sigma^k$ by
   $M$ to the symmetric realizability problem of
   $n =
   f(k)$ processes with a specification in LTL of size polynomial in
   $|Q| \cdot |\Gamma| \cdot |w|$.
  \end{lemma}
\begin{proof}
We build a specification that requires the processes to output the Turing tape configuration along an execution of the machine. The specification is realizable if and only if the Turing machine does \textbf{not} accept the word.
Every process outputs the value of one Turing tape cell and if the tape head is at the cell, also the state of the Turing machine. There are $n$ input signals to the architecture, and when the processes start, the left-most local input signals of the processes is used to tell one or more processes that the Turing tape computation should start at that cell with the tape head being initially there (with $w$ as the initial tape content). To account for the rotation-symmetry, the processes output not only the tape content and tape head position, but also the current boundaries of the tape. The specification is modeled such that if start and end markers collide, the simulation of the Turing machine can stop.

The specification also includes conjuncts that require all processes together to simulate the Turing machine computation correctly and to never reach an accepting state. Whenever the alternating Turing machine branches universally, the left-most local process input signal is used to select which successor state is picked. In case of existential branching, the processes can decide which successor state to pick. 
Enforcing the specification to be realizable if and only if the word $w$ is \textbf{not} accepted by the Turing machine helps with taking care of the diverging computations of the Turing machine and those computations that exceed the space bound. Both count as non-accepting in the definition of space-bounded Turing machines. Since these runs never visit accepting states and/or permit the simulation to stop, they are allowed to be simulated by a synthesized implementation.

The specification can be written with size polynomial in $|Q| \cdot |\Gamma| \cdot |w|$ as we only need to define the specification for one process. By the symmetry of the architecture, the other processes have to fulfill it as well.
\end{proof}
A more detailed proof can be found in the appendix.

  \begin{corollary}
   The rotation-symmetric realizability problem (for LTL) has a time complexity that is exponential in the number of processes.
  \end{corollary}
  \begin{proof}
   Given the question whether a word
   $w = w_0 \ldots
   w_{k-1}$ is in the language defined by some
   $(c+1)$-EXPTIME
   $=$ $(c)$-AEXPSPACE problem for some
   $c \in
   \NN$, we can reduce it to the symmetric realizability problem for an LTL specification of length polynomial in
   $k$ and with a number of processes that is 
   $(c)$-exponential in
   $k$.
   Since by the space hierarchy theorem \cite{Ranjan1991289}, the
   $(c)$-EXPTIME hierarchy is strict for increasing
   $c$, we can conclude that in general, we cannot solve the symmetric realizability problem faster than in time exponential in the number of components.
  \end{proof}

\section{The General Case -- Undecidability}
\label{sec:undecidability}

  The synthesis problem for standard, not necessarily symmetric, distributed systems is decidable as long as the processes can be ordered with respect to their relative knowledge about the system inputs~\cite{DBLP:conf/lics/FinkbeinerS05}.
  The problem becomes undecidable as soon as it contains an information fork, i.e., a pair of processes with incomparable knowledge. The simplest such architecture is Pnueli and Rosner's A0 architecture~\cite{DBLP:conf/focs/PnueliR90}, shown on the left in Fig.~\ref{fig:S0Architecture}.
  In this section, we show that for symmetric synthesis, even architectures without information forks, such as the S0 architecture shown on the right in Fig.~\ref{fig:S0Architecture}, are undecidable.
  Our proof is based on Pnueli and Rosner's undecidability argument for A0:
  
\begin{lemma}[\cite{DBLP:conf/focs/PnueliR90}]
  \label{lemma:pnueliroesner}
  For a given Turing machine $M$, there exists an LTL formula $\psi$ that is realizable in the distributed architecture A0 if and only if $M$ halts and such that the two processes of the unique implementation of $M$ sequentially output binary encodings of the configurations of the Turing machine on $y$ (or $z$, respectively) upon the first $\TRUE$ value on the input $u$ (or $v$, respectively).
\end{lemma}

  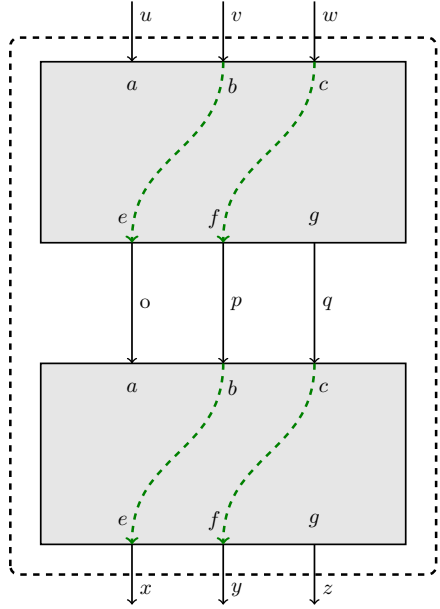
\begin{figure}
   \centering%
   \scalebox{0.8}{\begin{tikzpicture}
    \draw[very thick, dashed, rounded corners] (-0.5,-5.5) rectangle (6.5,3.4);
    \draw[thick,fill=black!10!white] (0,0) rectangle +(6,3); \draw[thick,fill=black!10!white] (0,-5) rectangle +(6,3);
    
    	\draw[dashed,very thick,->,color=green!50!black] (3,3) .. controls +(0,-1.5) and +(0,1.5) .. (1.5,0);
   		\draw[dashed,very thick,->,color=green!50!black] (4.5,3) .. controls +(0,-1.5) and +(0,1.5) .. (3,0);
   		\draw[dashed,very thick,->,color=green!50!black] (3,-2) .. controls +(0,-1.5) and +(0,1.5) .. (1.5,-5);
   		\draw[dashed,very thick,->,color=green!50!black] (4.5,-2) .. controls +(0,-1.5) and +(0,1.5) .. (3,-5);
   		
    \draw[thick,->] (1.5,4) -- node[near start,right] {$u$} (1.5,3); \draw[thick,->] (3,4) -- node[near start,right] {$v$} (3,3); \draw[thick,->] (4.5,4) -- node[near start,right] {$w$} (4.5,3);
    \draw[thick,->] (1.5,-5) -- node[near end,right] {$x$} (1.5,-6); \draw[thick,->] (3,-5) -- node[near end,right] {$y$} (3,-6); \draw[thick,->] (4.5,-5) -- node[near end,right] {$z$} (4.5,-6);
    \draw[thick,->] (1.5,0) -- node[right] {o} (1.5,-2); \draw[thick,->] (3,0) -- node[right] {$p$} (3,-2); \draw[thick,->] (4.5,0) -- node[right] {$q$} (4.5,-2);
    \node at (1.5,2.6) {$a$}; \node at (3.15,2.6) {$b$}; \node at (4.65,2.6) {$c$}; \node at (1.35,0.4) {$e$}; \node at (2.85,0.4) {$f$}; \node at (4.5,0.4) {$g$};
    \node at (1.5,-2.4) {$a$}; \node at (3.15,-2.4) {$b$}; \node at (4.65,-2.4) {$c$}; \node at (1.35,-4.6) {$e$}; \node at (2.85,-4.6) {$f$}; \node at (4.5,-4.6) {$g$};
    
   \end{tikzpicture}}
   \caption{Symmetric architecture S2. The symmetric synthesis problem for S2 is undecidable. The dashed arrows in the process boxes show how the specification given in the proof of Lemma~\ref{lem:undecidabilityOfS2} requires the processes to forward the local input streams. }
   \label{fig:S2Architecture}
  \end{figure}

  Because of the undecidability of the halting problem, Lemma~\ref{lemma:pnueliroesner} means that the distributed synthesis problem of architecture A0 is undecidable.
We prove the undecidability of the symmetric synthesis problem of architecture S0 in two steps. First, we establish the undecidability of the larger architecture S2, depicted in Figure~\ref{fig:S2Architecture}, by showing that the realizability of $\psi$ in A0 can be reduced to the symmetric realizability of an LTL formula over S2; in the second step, we encode the synthesis problem of S0 into the synthesis problem of S2 and thus establish that the synthesis problem for the simpler architecture S0 is undecidable as well.

\begin{lemma}
\label{lem:undecidabilityOfS2}
  The symmetric synthesis problem for architecture S2 is undecidable.
\end{lemma}

\begin{proof}
We show that there exists an implementation for the specification
   $\psi$ in the A0 architecture if and only if there exists a joint implementation for the two processes in the S2 architecture that satisfies
   $\psi' = \psi_d \wedge \LTLG ( v \leftrightarrow \LTLX o ) \wedge \LTLG ( w \leftrightarrow \LTLX p
   )$, where
   $\psi_d$ results from prefixing all occurrences of the signals
   $y$ and
   $z$ in
$\psi$ with a next-time operator.

The results of the two synthesis problems can be translated into each other.
A distributed implementation of $\psi$ over A0 is necessarily symmetric:
both processes output the same bitstream when reading a
   $\TRUE$ value as their local input for the first time. To obtain an implementation for S2, we simulate the process with input 
   $a$ and use
   $g$ as the local output.
   Additionally, we copy all values from
   $b$ to
   $e$, and
   $c$ to
   $f$.

Conversely, an implementation found by the symmetric synthesis of S2 provides an implementation of $\psi$ in A0. The key property of the architecture S2 is that the process does not know if the local input
   $b$ is the (delayed)
   $a$ input to the other process, or if its
   $c$ input is the (Turing machine tape) output of the other process. Thus, it cannot find out if it is the top process or the bottom process in the architecture and must prevent violating the specification in either case. A more detailed proof is given in the appendix.
\end{proof}

In order to reduce the symmetric synthesis problem of S2 to the symmetric synthesis problem of S0, we introduce compression functions that time-share multiple signals of S2 into a single signal in S0.

   Let
   $\AP$ be a set of signals.
   We call a function
   $f : (2^\AP)^\omega \rightarrow (2^{\{ \chi
    \}})^\omega$ for some Boolean variable $\chi$ a \newterm{compression function} if
   $f$ is injective.
   We call a function
   $f'$ that maps a specification over the signal set
   $\AP$ to a different specification over the signal set
   $\{ \chi
    \}$ the \newterm{adjunct compression function} to
   $f$ if for all
   $w \in
   (2^\AP)^\omega$ and specifications
   $\psi$ over
   $\AP$, we have that
   $w \models
   \psi$ if and only if
   $f(w) \models
   f'(\psi)$.

   In the appendix, we give such a pair of compression functions for
   LTL.  The compression mechanism 
  is illustrated in Figure~\ref{fig:ltlWordCompression}.
  One clock cycle in the four-bit-per-character version of a word is spread to 10 computation cycles in the one-bit-per-character version of the word.
  Every 10 cycles, the 2-cycle \newterm{character start sequence} (CSS)
  $\{ \chi \} \{ \chi
   \}$ is instantiated, followed by four two-cycle slots for every signal in
  $\AP$.
  Note that the construction ensures that whenever we have
  $\{ \chi \} \{ \chi \}
  \emptyset$ as a part in a compressed word, then we know that a character start sequence begins on the first occurrence of
  $\{ \chi
   \}$ in this part.
  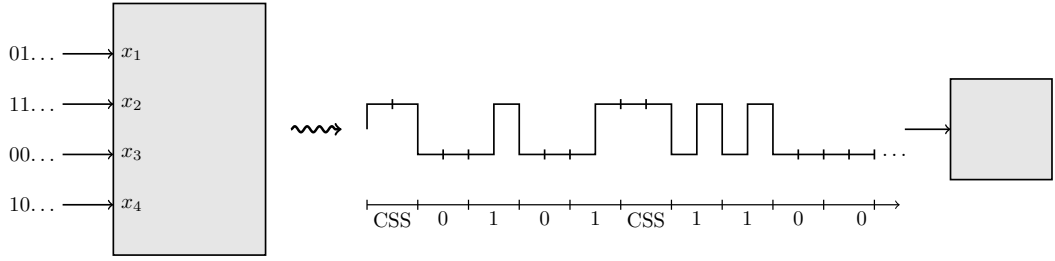
\begin{figure*}
   \resizebox{\linewidth}{!}{
    \begin{tikzpicture}[scale=0.85]
     \draw[thick,fill=black!10!white] (0,5) rectangle +(3,-5); \draw[thick,->] (-1,4) node[left]
     {01$\ldots$} -- (0,4) node[right]
     {$x_1$} ; \draw[thick,->] (-1,3) node[left]
     {11$\ldots$} -- (0,3) node[right]
     {$x_2$} ; \draw[thick,->] (-1,2) node[left]
     {00$\ldots$} -- (0,2) node[right]
     {$x_3$} ; \draw[thick,->] (-1,1) node[left]
     {10$\ldots$} -- (0,1) node[right]
     {$x_4$} ;

     \draw [very thick, ->, decorate, decoration={snake,amplitude=.5mm,segment length=2.07mm}] (3.5,2.5) -- (4.5,2.5);

     \draw[thick] (5,2.5) -- (5,3) -- (6,3) -- (6,2) -- (7.5,2) -- (7.5,3) -- (8,3) -- (8,2) -- (9.5,2) -- (9.5,3) -- (11.0,3) -- (11.0,2) -- (11.5,2) -- (11.5,3) -- (12,3) -- (12,2) -- (12.5,2) -- (12.5,3) -- (13,3) -- (13,2) -- (15,2) node[right]
     {$\ldots$};

     \draw[thick] (5.5,3.1) -- (5.5,2.9); \draw[thick] (6.5,2.1) -- (6.5,1.9); \draw[thick] (7,2.1) -- (7,1.9); \draw[thick] (8.5,2.1) -- (8.5,1.9); \draw[thick] (9,2.1) -- (9,1.9); \draw[thick] (10.0,3.1) -- (10.0,2.9); \draw[thick] (10.5,3.1) -- (10.5,2.9); \draw[thick] (13.5,2.1) -- (13.5,1.9); \draw[thick] (14.0,2.1) -- (14.0,1.9); \draw[thick] (14.5,2.1) -- (14.5,1.9); \draw[thick] (15.0,2.1) -- (15.0,1.9);

     \draw[thick,fill=black!10!white] (16.5,1.5) rectangle +(2,2); \draw[thick,->] (15.6,2.5) -- (16.5,2.5);

     \draw[semithick,->] (5,1) -- node[below] {CSS} (6,1) -- node[below] {0} (7,1) -- node[below] {1} (8,1) -- node[below] {0} (9,1) -- node[below] {1} (10,1) -- node[below] {CSS} (11,1) -- node[below] {1} (12,1) -- node[below] {1} (13,1) -- node[below] {0} (14,1) -- node[below] {0} (15.5,1); \draw[semithick] (5,0.9) -- (5,1.1); \draw[semithick] (6,0.9) -- (6,1.1); \draw[semithick] (7,0.9) -- (7,1.1); \draw[semithick] (8,0.9) -- (8,1.1); \draw[semithick] (9,0.9) -- (9,1.1); \draw[semithick] (10,0.9) -- (10,1.1); \draw[semithick] (11,0.9) -- (11,1.1); \draw[semithick] (12,0.9) -- (12,1.1); \draw[semithick] (13,0.9) -- (13,1.1); \draw[semithick] (14,0.9) -- (14,1.1); \draw[semithick] (15,0.9) -- (15,1.1);
    \end{tikzpicture}
   }
   \caption{An example for compressing a word with
    $|\AP| =
    4$.}
   \label{fig:ltlWordCompression}
  \end{figure*}

  \begin{theorem}\label{theorem:undecidabilityS2}
  The symmetric synthesis problem for architecture S0 is undecidable.
\end{theorem}

\begin{proof}
In order to reduce the symmetric synthesis problem of architecture $S0$ to the symmetric synthesis problem of architecture S2, we compress $u,v,w$ into signal $x$; $o,p,q$ into signal $y$; and $x,y,z$ into signal $z$. A more detailed proof is given in the appendix.
\end{proof}

\section{Conclusions}

In this paper, we have studied the problem of synthesizing symmetric
systems.  Our new synthesis algorithm is a useful tool in the development of
distributed algorithms, because it checks automatically if
certain properties in a design problem require symmetry breaking.  

Our algorithm synthesizes implementations of rotation-symmetric architectures, i.e., architectures where the processes observe all inputs. The undecidability result for the architecture S0 indicates that it is impossible to extend the synthesis algorithm to architectures where the processes no longer have access to all inputs.
A promising direction of research,
however, is to use our results to extend existing \emph{semi-algorithms} for
synthesis under incomplete information to such symmetric architectures.
An example for such an approach is \emph{bounded
  synthesis}~\cite{DBLP:conf/atva/ScheweF07a}, which determines if
there exists an implementation with at most $n$ states, where $n$
is a given bound. The specification is translated into a universal
co-B\"uchi automaton, which is then, together with the bound $n$,
encoded into a \newterm{satisfiability modulo theory} problem. To
ensure correctness under incomplete information, constraints are added
that ensure that if a process cannot distinguish two inputs, it
transitions to the same successor state. Similarly, for symmetric
synthesis, constraints can be added that ensure that the outputs of
the individual processes are identical in states that are
indistinguishable for them.

Algorithms for symmetric synthesis procedures also offer a new
perspective on the problem of synthesizing arbitrarily scalable
(i.\,e.\ \newterm{parametric}) systems. Due to the undecidability of
the problem, only very limited solutions to this problem have been
found so far. For example, 
Jacobs and Bloem~\cite{DBLP:journals/corr/JacobsB14} tackle the case of asynchronous processes with local input in a ring architecture and use the bounded synthesis approach mentioned above.
Emerson and Srinivasan
\cite{DBLP:conf/podc/EmersonS90} present a solution for a
multi-process version of a small subset of the temporal logic CTL
while Attie and Emerson \cite{DBLP:journals/toplas/AttieE98} give a
different solution allowing a bigger subset of CTL but only
guaranteeing correctness of the solution if certain other conditions
are fulfilled, like the dead-lock freeness of the solution
produced. In such a setting, symmetric synthesis can be used to detect
specifications that are unrealizable even for small system sizes -- if
there is no solution for a fixed number of processes $n$, then there
is certainly none for scalable systems as well.

\bibliography{bib}

\clearpage
\appendix

\section{Appendix -- Proof Details}

\subsection{Additional Preliminaries}

We use \newterm{Moore machines} as finite-state model for regular computation trees.
Formally, a Moore machine is a tuple $\mathcal{M} = (S,\mathcal{I},O,\delta,s_\mathit{init},L)$ with the (finite) set of states $S$, the input alphabet $\mathcal{I}$, the output alphabet $O$, the initial state $s_\mathit{init} \in S$, and the labelling function $L : S \rightarrow O$. 
A Moore machine induces a computation tree $\langle T, \tau \rangle$ with $T = \mathcal{I}^*$ and $\tau : T \rightarrow O$ such that for all $t_0 \ldots t_n \in T$, we have that $\tau(t_0 \ldots t_n) = L(\delta(\delta(\ldots(\delta(\delta(s_\mathit{init},t_0),t_1),\ldots), t_{n-1}),t_n))$. Moore machines induce regular computation trees, i.e., computation trees that only have a finite number of distinct sub-trees.

Given a Moore machine, an \newterm{extended computation tree} induced by it is the same as a computation tree induced by the Moore machine, except that the tree labels are in $S \times O$, where for every node $t$, the first label element of $\tau(t)$  describes the state of the Moore machine after reading the input $t$ from the initial state, and the second label element describes the last output after reading $t$ from the initial state as before.
\subsection{Additional Definitions}

In Definition~\ref{def:symmetricSynthesisProblem}, we used the standard definition of parallel composition to say what it means to plug a process implementation into a symmetric architecture. For the sake of completeness, let us formally define this special case of parallel composition.

  \begin{definition}
   \mylabel{def:AggregatedMooreMachine} Given an architecture
   $\mathcal{E} =
   (S,P,\AP^I_G,E^\mathit{in},E^\mathit{out})$ for some process interface
   $\mathcal{N} =
   (\AP^I,\AP^O)$ and some Moore machine
   $\mathcal{M} =
   (Q,I,O,\delta,q_0,L)$ with
   $I =
   2^{\AP^I}$ and
   $O =
   2^{\AP^O}$, we define the \newterm{aggregated Moore machine} of the architecture and $\mathcal{M}$ as
   $\mathcal{M}' =
   (Q',I',O',\delta',q'_0,L')$ with:
   \begin{itemize}
   \item
    $Q' = (P \rightarrow
    Q)$,
   \item
    $I' =
    2^{\AP^I_G}$,
   \item
    $O' =
    2^{S}$,
   \item
    for all
    $f \in
    Q'$, we have
    $L'(f) = \{ s \in S \mid \exists (p,x) \in P \times \AP^O: E^\mathit{out}(p,x) = s,\ x \in L(f(p))
     \}$, 
   \item
    for all
    $f \in
    Q'$ and
    $X \subseteq
    \AP_G^I$,
    $\delta'(f,X) =
    f'$ such that for all
    $p \in
    P$,
    $f'(p) = \delta(f(p),\{ x \in \AP^I_L \mid E^\mathit{in}(p,x) \in (X \uplus L(f)
     )\})$, and
   \item
    for all
    $p \in
    P$,
    $q_0(p) =
    q_0$.
   \end{itemize}
  \end{definition}
  This definition ensures that the values of all signals are ``exported'' from the aggregated finite-state machine.
  Thus, when specifying the system behaviour of an aggregated system in a language such as \newterm{linear-time temporal logic} (LTL), we can refer to the signals used internally between the components.

In the main part of the paper, we also define computation trees that encode the behavior of a rotation-symmetric architecture after we plug one process into it. If the process is a finite-state machine, then the resulting computation tree for the behavior of the complete architecture is regular, and hence can be translated (back) to a Moore machine. We call this Moore machine for the behavior of the complete rotation-symmetric architecture implementation the \newterm{symmetric product} of the single process, whose definition we give next. The reader is reminded that $O$ and $\mathcal{O}$ are defined on page~\pageref{def:mathcalO}.

  \begin{definition}[Symmetric product]
   \label{def:SymmetricProduct} Given a Moore machine
   $\mathcal{M} =
   (S,\mathcal{I},O,\delta,s_0, \allowbreak L)$, we say that a Moore machine
   $\mathcal{M}' =
   (S',\mathcal{I},\mathcal{O},\delta',s'_0,L')$ is the symmetric product of
   $\mathcal{M}$ if
   $S' =
   S^n$,
   $s'_0 =
   (s_0)^n$, and for all
   $s''_0, \ldots, s''_{n-1} \in
   S$,
   $(i_0, \ldots, i_{n-1}) \in
   \mathcal{I}$:
   \begin{equation*}
    \delta'((s''_0,\ldots,s''_{n-1}),(i_0,\ldots,i_{n-1})) = (s'''_0,\ldots,s'''_{n-1}) 
   \end{equation*}
   s.\,t.\
   $\forall 0 \leq j < n: s'''_j = \delta(s''_j,\rot((i_0,\ldots,i_{n-1}),-j))
   $ and
   $L'(s''_0, \allowbreak \ldots, \allowbreak s''_{n-1}) = (L(s''_0), \allowbreak L(s''_1), \allowbreak \ldots, \allowbreak
   L(s''_{n-1}))$.
  \end{definition}
  Note that Definition \ref{def:SymmetricProduct} is just a combination of Definition~\ref{def:rotationSymmetricArchitecture} and the usual definition of parallel composition of Moore machines, applied to architectures consisting of a single cycle of processes.

\subsection{Proof of the Symmetry Lemma}

Let in the following for every $i \in \NN$ the expression $\mathcal{O}_i$ denote the local output of process $i$, i.e., let us define $\mathcal{O}_i = 2^{\AP^O \times \{i\}}$.

\allowdisplaybreaks{}

\begin{proof}
   $\Leftarrow$: The fact that the computation tree induced by the symmetric product of some Moore machine has the symmetry property follows directly from the definitions.

   $\Rightarrow$: For the converse direction, we prove that from every regular computation tree with the symmetry property, we can construct a Moore machine that is an implementation for one process, and by taking the symmetric product of the Moore machine, we obtain a product machine whose computation tree is in turn the one that we started with.

   Let
   $\langle T, \tau
   \rangle$ be the computation tree to start with.
   As it is regular, we have an equivalence relation over the nodes in the tree.
   Let
   $[ \cdot
   ]$ be the function that maps a tree node in
   $t$ onto a tree node representing its equivalence class, so for all
   $t, t' \in
   T$, we have that the sub-trees induced by
   $t$ and
   $t'$ are the same if and only if
   $[t] =
   [t']$, and for every
   $t$ there is some
   $t'$ such that
   $[t] =
   t'$.
   We build a Moore machine for one process in the symmetric architecture from
   $\langle T, \tau
   \rangle$ by setting
   $\mathcal{M} =
   (S,\mathcal{I},O,\delta,s_\mathit{init},L)$ with:
   \begin{eqnarray*}
    S & = & \{ [t] \mid t \in T \} \\
    \delta(s,x) &
    = &
    [sx] \text{ for all } x \in \mathcal{I} \text{ and } s \in S \\
    s_\mathit{init} &
    = &
    [\epsilon] \\
    L(s) &
    = &
    \tau(s)|_{O_0} \text{ for all } s \in S
   \end{eqnarray*}
   We now show that the symmetric product of
   $\mathcal{M}$ induces a computation tree that is the same as
   $\langle T, \tau
   \rangle$.
   If we take the symmetric product (Definition \ref{def:SymmetricProduct}) of
   $\mathcal{M}$, we obtain
   $\mathcal{M}' =
   (S',\mathcal{I},\mathcal{O},\delta',s'_\mathit{init},L')$ with:
   \begin{eqnarray*}
    S' & = &
    \{ ([t_0], \ldots, [t_{n-1}]) \mid t_0, \ldots, t_{n-1} \in T \} \\
    \delta'((t_0,\ldots,t_{n-1}),x) &
    = &
    ([t_0 \rot(x,0)], [t_1 \rot(x,-1)], \ldots, [t_{n-1} \rot(x,-n+1)]) \\
    s_\mathit{init} &
    = &
    ([\epsilon],\ldots,[\epsilon]) \\
    L'((t_0,\ldots,t_{n-1})) &
    = &
    (\tau(t_0)|_{O_0},\tau(t_1)|_{O_0},\ldots,\tau(t_{n-1})|_{O_0})
   \end{eqnarray*}
   Let
   $\langle T', \tau'
   \rangle$ be the extended computation tree induced by
   $\mathcal{M}'$ with
   $\tau' : T' \rightarrow S' \times
   \mathcal{O}$.
   We can show by induction that for every
   $t \in
   T$, we have that
   $\tau'(t)|_{S'} = ([\rot(t,0)], \allowbreak{} [\rot(t,-1)], \allowbreak{} [\rot(t,-2)], \allowbreak{} \ldots,
   [\rot(t,-n+1)])$.
   The induction basis is trivial, as
   $\tau'(t)|_{S'} ( \epsilon ) = ([\epsilon],[\epsilon],[\epsilon], \ldots,
   [\epsilon])$.
   For the inductive step, we have:
   \begin{eqnarray}
   \label{stepA1} & & \tau'(tx)|_{S'}   \\
   \label{stepA2} &
    = &
    ([t_0 \rot(x,0)], [t_1 \rot(x,-1)], \ldots, [t_{n-1} \rot(x,-n+1)])  \\
    \nonumber &
    &
    \text{for } \tau'(t)|_{S'} ( t ) = (t_0,t_1,\ldots,t_{n-1}) \\
    \label{stepA3} &
    = &
    ([[\rot(t,0)] \rot(x,0)], \ldots, [[\rot(t,-n+1)] \rot(x,-n+1)]) \\
    \label{stepA4} &
    = &
    ([\rot(t,0) \rot(x,0)], \ldots, [\rot(t,-n+1) \rot(x,-n+1)]) \\
     &
    = &
    ([\rot(tx,0)], \ldots, [\rot(tx,-n+1)])
   \end{eqnarray}
	In step (\ref{stepA1})-(\ref{stepA2}) of this deduction, we applied the definitions of the elements of $\mathcal{M}$ and $\mathcal{M}'$.
   In step (\ref{stepA2})-(\ref{stepA3}), we used the inductive hypothesis.
   In step (\ref{stepA3})-(\ref{stepA4}), we used the regularity of the tree: for some
   $t \in
   T$ and
   $x \in
   \mathcal{I}$, we need to have
   $[[t]x] =
   [tx]$ as the subtree induced by
   $[t]x$ has to be the same as the one induced by
   $tx$, as otherwise
   $[t]$ and
   $t$ would not be in the same equivalence class of subtrees (which is a contradiction).
   The last step uses the fact that if we concatenate two strings that are rotated by the same number of indices, then we can also first concatenate, and then rotate.

   Now let us have a look at the outputs in the extended computation tree
   $\langle T', \tau'
   \rangle$.
   For every
   $t \in
   T'$, we have:
   \begin{eqnarray}
    & & \tau'(t)|_{\mathcal{O}} \\
    &
    = &
    L'(\tau'(t)|_{S'}) \\
    \label{stepB1} &
    = &
    L'([\rot(t,0)], \ldots, [\rot(t,-n+1)]) \\
    \label{stepB2} &
    = &
    (\tau([\rot(t,0)])|_{\mathcal{O}_0}, \ldots, \tau([\rot(t,-n-1)])|_{\mathcal{O}_0}) \\
    \label{stepB3} &
    = &
    (\tau(\rot(t,0))|_{\mathcal{O}_0}, \ldots, \tau(\rot(t,-n-1))|_{\mathcal{O}_0}) \\
    \label{stepB4} &
    = &
    (\tau(t)|_{\mathcal{O}_0}, \tau(t)|_{\mathcal{O}_1},\ldots, \tau(t)|_{\mathcal{O}_{n-1}}) \\
    \label{stepB5} &
    = &
    \tau(t) 
   \end{eqnarray}
   In step (\ref{stepB1})-(\ref{stepB2}), we simply applied the definition of
   $L'$.
   In step (\ref{stepB2})-(\ref{stepB3}), we used the fact that we are dealing with equivalence classes over nodes in the computation tree
   $\langle T, \tau
   \rangle$ that respect the labelling of the system.
   In step (\ref{stepB3})-(\ref{stepB4}), we use the symmetry property of
   $\langle T, \tau
   \rangle$.
   For every
   $i \in
   \{0,\ldots,n-1\}$, we have
   $\tau(\rot(t,i))|_{\mathcal{O}_0} =
   \rot(\tau(t),i)|_{\mathcal{O}_0}$ by this property, and then
   $rot(\tau(t),i)|_{\mathcal{O}_0} =
   \tau(t)|_{\mathcal{O}_i}$ by renaming.
   In the last step, we just plug together the tuple.
\end{proof}

\subsection{Correctness of the $\rep_S$ Function}

  The definition of the
  $\reps$ function in Section~\ref{subsec:definitionOfReps} is supposed to describe how to compute the symmetry degree of a word, i.e., the number of processes getting the same rotations of an input proposition valuation or the number of rotations of the output of the processes that lead to the same element of $\mathcal{O}$.
  We prove that the definition of the
  $\reps$ function achieves this goal in two steps and start with the following sub-lemma:

  \begin{lemma}
   \label{RepSLemma1} If there are precisely
   $m$ values
   $j \in
   \{0,\ldots,n-1\}$ (for some
   $m \in
   \NN$) such that
   $\rot(t,j) =
   t$ for some
   $t \in
   \mathcal{I}^*$, then the list of indices
   $L = \{\frac{0 \cdot n}{m}, \frac{1 \cdot n}{m}, \ldots,
    $ $\frac{(m-1) \cdot
     n}{m}\}$ is precisely the list of indices
   $\geq
   0$ but
   $<
   n$ such that for all
   $l \in
   L$, we have
   $\rot(t,l) =
   t$ but for all
   $l' \notin
   L$, we have
   $\rot(t,l') \neq
   t$ or either
   $l'<0$ or
   $l'\geq
   n$.
  \end{lemma}
  \begin{proof}
   For all
   $j,j' \in
   L$, we know that
   $j+j' \in
   L$ as well since for all
   $t \in
   \mathcal{I}^*$,
   $\rot(t,{j+j'}) = \rot(\rot(t,j'),j) = \rot(t,j') =
   t$.
   Furthermore,
   $\rot(t,n)=t$.

   To show that all elements in
   $L \cup
   \{n\}$ are equally spaced (modulo
   $n$), consider the converse.
   So we have
   $0 \leq l<l'<l'' <
   n$ with
   $l'-l \neq
   l''-l'$ and there are no indices in
   $L$ in between
   $l'$ and
   $l$ or
   $l''$ and
   $l'$, respectively.
   By the argument above if
   $l'-l <
   l''-l'$ we also have
   $l'+(l'-l) \in
   L$ or if
   $l'-l >
   l''-l'$ we also have
   $l+(l''-l') \in
   L$, which is a contradiction.
   The case that involves wrapping around in the modulo space can be proven similarly.

   So we know that there are
   $m$ equally spaced elements in
   $L$, and by the same line of reasoning, we can also deduce that the spacing between the elements in $L$ is the same as the spacing between $n$ and the largest element in L.
   Since furthermore
   $\rot(t,0) = t$ and
   $\rot(t,n) = t$ for all
   $t \in
   \mathcal{I}^*$, the claim follows.
  \end{proof}
  Lemma \ref{RepSLemma1} can alternatively be shown by applying a theorem by Fine and Wilf \cite{Fine1965} on the combinatorics on words.
  To use it, we would however have to rearrange the letters in a word, and describing that construction would be more complicated than giving a direct proof, which is why the latter has been done here.
  \begin{lemma}
   \label{lem:repsNofRotations} For every
   $t_0 \ldots t_{k-1} \in \mathcal{I}^k, k \in
   \NN$, we have
   \begin{equation*}
    \reps(t_0 \ldots t_{k-1}) = |\{j \in \{0,\ldots,n-1\} : \rot(t_0 \ldots t_{k-1},j)= t_0 \ldots t_{k-1} \}|
   \end{equation*}
  \end{lemma}
  \begin{proof}
   The proof is done by induction on the length of
   $t$.

   \textbf{Basis:} Trivial, since
   $\rot(\epsilon,j)=\epsilon$ for every
   $j \in
   \ZZ$.

   \textbf{Inductive step:} Assume that the number of neutral rotations for
   $t_0 \ldots
   t_{k-2}$ is
   $m$ (we denote those rotation values
   $j$ of
   $t_0 \ldots
   t_{k-2}$ to be \newterm{neutral} for which
   $\rot(t_0 \ldots t_{k-2},j) = t_0 \ldots
   t_{k-2}$) and the number of neutral rotations for
   $t_0 \ldots
   t_{k-1}$ is
   $m'$ .
   By the inductive hypothesis,
   $\reps(t_0 \ldots
   t_{k-2})=m$.

   Clearly,
   $m'$ is a divisor of
   $m$ since otherwise there exists a
   $y \in \{\frac{0 \cdot n}{m'}, \ldots,
    $ $\frac{(m'-1) \cdot
     n}{m'}\}$ such that
   $\rot(t_0 \ldots t_{k-2},y) \allowbreak \neq t_0 \ldots
   t_{k-2}$, so:
   \begin{equation*}
    \rot(t_0 \ldots t_{k-1},y) = \rot(t_0 \ldots t_{k-2},y) \rot(t_{k-1},y) \neq t_0 \ldots t_{k-2} \rot(t_{k-1},y) = t_0 \ldots t_{k-1}
   \end{equation*}
   Analogously,
   $m'$ is a divisor of
   $\rep(t_k)$. In both cases, we would otherwise get a contradiction with Lemma~\ref{RepSLemma1}.

   On the other hand, for every
   $m'$ that is a divisor of
   $m$ and
   $\rep(t_k)$, we have for all
   $y \in \{\frac{0 \cdot n}{m'}, \ldots, \frac{(m'-1) \cdot
     n}{m'}\}$:
   \begin{align*}
    & \rot(t_0 \ldots t_{k-1},y) \\ 
    =\ &
    \rot(t_0 \ldots t_{k-2},y) \rot(t_{k-1},y) \\
    =\ &
    t_0 \ldots t_{k-2} \rot(t_{k-1},y) & \text{by Lemma \ref{RepSLemma1}}\\ 
    =\ &
    t_0 \ldots t_{k-2} t_{k-1} &  \text{by Lemma \ref{RepSLemma1}}\\
    =\ &
    t_0 \ldots t_{k-1}
   \end{align*}
   Clearly, the greatest common divisor of
   $m$ and
   $\rep(t_k)$ is the (unique) greatest such number, therefore
   $\reps(t_0 \ldots
   t_{k-1})=m'$.
  \end{proof}

\subsection{Proof of Lemma~\ref{lem:symmetryLemma2} (The Second Symmetry Lemma)}

To keep the presentation of the following proof concise, we need to give a name to the normalization function that maps all input streams that can be unified by rotation onto the same input stream.
  For all
  $t \in
  \mathcal{I}^*$, we define:
  \begin{equation*}
   \eta_S(t) = \min_{i \in \{ 0, \ldots, n-1 \}} \rot(t,i)
  \end{equation*}

  \begin{definition}[Symmetric completion]
   Let
   $\langle T, \tau
   \rangle$ be a computation tree with
   $T =
   \mathcal{I}^*$ and
   $\tau : T \rightarrow
   \mathcal{O}$.
   We call a tree
   $\langle T, \tau'
   \rangle$ with
   $\tau' : T \rightarrow
   \mathcal{O}$ a symmetric completion of
   $\langle T, \tau
   \rangle$ if for all
   $t \in
   T$, we have
   $\tau'(t) =
   \rot(\tau(\eta_S(t)),i)$ for some
   $i \in
   \NN$ with
   $\rot(\eta_S(t),i) =
   t$.
  \end{definition}
We are now ready to discuss the proof of Lemma~\ref{lem:symmetryLemma2}.
\begin{proof}
   The first part of the claim follows directly from the definition of the symmetric completion. For all $t \in T$, $i \in \NN$, we have
   \begin{eqnarray*}
   &  & \tau(\rot(t,i)) \\
   &= & \tau(\rot(\eta_S(t),i+j)) \\
   &= & \rot(\tau(\eta_S(t)),i+j) \\
   &= & \rot(\tau(t),i)
   \end{eqnarray*}%
   for $j = \min_{j' \in \NN} \rot(r,-j')$ if the symmetric completion actually exists. The symmetric completion exists if and only if having 
   $\tau'(t) =
   \rot(\tau(\eta_S(t)),i)$ for some
   $i \in
   \NN$ with
   $\rot(\eta_S(t),i) =
   t$ for every $t \in T$ is possible. The only way for this property to be unfulfillable is if for some $t \in T$ and $i,j \in \NN$, we have $\rot(t,i) = \rot(t,j)$ but $\tau(\rot(t,i)) \neq \tau(\rot(t,j))$. Equivalently we can ask if there is the possibility to have $t = \rot(t,j-i)$ but $\tau(t) \neq \tau(\rot(t),j-i)$. Since we require $\langle T, \tau \rangle$ to have the property that for every $t \in T$, we have $\rep_S(t) \divides \rep_S(\tau(t))$ and by Lemma~\ref{RepSLemma1}, the neural rotations are evenly spaced, we can see that $t = \rot(t,j-i)$ and $\tau(t) \neq \tau(\rot(t,j-i))$ cannot hold at the same time.
   
   For the second part, we assume that we are given some finite-state machine $\mathcal{M} =
   (S,\mathcal{I},\mathcal{O},s_\mathit{init},\delta,L)$ that induces the computation tree $\langle T, \tau \rangle$. Since the classes of regular trees and finite-state (Moore) machines $\mathcal{M}$ are isomorphic, this assumption comes without loss of generality.
   
   We prove the regularity of $\langle T, \tau'\rangle$ by building a finite-state machine
   $\mathcal{M}' =
   (S',\mathcal{I},\mathcal{O}, \allowbreak{} s'_\mathit{init}, \allowbreak{} \delta', \allowbreak{} L')$ that induces $\langle T, \tau' \rangle$.
   
   The states in
   $\mathcal{M}$ are the equivalence classes of nodes in
   $\langle T, \tau
   \rangle$.
   Without loss of generality, let us assume that elements
   $t,t' \in
   T$ such that
   $\rep_S(t) \neq
   \rep_S(t')$ are never put into the same equivalence class.
   Note that this can only blow up the number of equivalence classes by a factor of at most
   $n$.\footnote{In principle, forcing two tree nodes to not be in the same equivalence class if they do not have some common property can make the number of equivalence classes infinite.
    This is however not the case here as for every
    $t\in T$ and
    $x \in
    \mathcal{I}$,
    $\rep_S(tx)$ can be computed from
    $\rep_S(t)$ and
    $x$.
    This way, computing the
    $\rep_S$ values of the prefix input stream can be done on-the-fly by a finite-state machine while reading the input, and when we compute its product with
    $\mathcal{M}$, the set of states of the resulting finite-state machine is then finite and serves as set of equivalence classes for the tree
    $\langle T, \tau
    \rangle$ such that no tree nodes with different
    $\rep_S$ values are put into the same equivalence class.} Furthermore assume that the representative element chosen for every equivalence class is a normalized prefix input whenever possible (again, without loss of generality).
   Thus, for every
   $s \in S$ representing a normalized input sequence and
   $i \in
   \mathcal{I}$,
   $\min_{u \in \NN}
   \rot(si,u)$ is normalised as well, and so is then
   $[\min_{u \in \NN}
   \rot(si,u)]$.
   We construct
   $\mathcal{M}'$ as follows:
   \begin{eqnarray*}
    S' & = & S \times \{0, \ldots, -n-1\} \\
    s'_\mathit{init} &
    = &
    (s_\mathit{init},0) \\
    \delta'((s,k),i) &
    = &
    ([\min_{j \in \NN} \rot(\rot(s,k) i,j)], - \argmin_{j \in \NN} \rot(\rot(s,k) i,j)) \\
    L'(s,k) &
    = &
    \rot(L(s),k)
   \end{eqnarray*}
   Let us now prove that
   $\mathcal{M}'$ represents the symmetric completion of
   $\langle T, \tau
   \rangle$.
   For this, it suffices to show that for every
   $t \in
   T$, we have
   $\delta'(s_\mathit{init},t) = ([\eta_S(t)],- \argmin_{j \in \NN} \allowbreak
   \rot(t,j))$.   
   By the definition of
   $L'$, we then have that
   $L'(\delta'(s_\mathit{init},t)) = \rot(L(\eta_S(t)), \allowbreak -\argmin_{j \in \NN}
   \rot(t,j))$ (as
   $L([\eta_S(t)]) = L(\eta_S(t))
   $), which corresponds to the symmetric completion of
   $\langle T, \tau
   \rangle$.

   Let us show that we have
   $\delta'(s_\mathit{init},t) = ([\eta_S(t)],- \argmin_{j \in \NN}
   \rot(t,j))$ for every
   $t \in
   T$ by induction.
   The induction basis for
   $t =
   \epsilon$ is trivial.
   For the inductive step, we have (for
   $t \in
   T$ and
   $x \in
   \mathcal{I}$):
   \begin{eqnarray}
   \!\!\!\! \!\!\!\! \!\!\!\! \delta'(s_\mathit{init},tx) & = &
    ([\min_{j \in \NN} \rot(\rot([\eta_S(t)],- \argmin_{j \in \NN} \rot(t,j)) x,j)], \nonumber \\
    &
    &
    \quad - \argmin_{j \in \NN} \rot(\rot([\eta_S(t)],- \argmin_{j \in \NN} \rot(t,j)) x,j)) \mylabel{eqnlineb1} \\
    &
    = &
    ([\min_{j \in \NN} \rot(\rot( [\min_{l \in \NN} \rot(t,l)],- \argmin_{j \in \NN} \rot(t,j)) x,j)], \nonumber \\
    &
    &
    \quad - \argmin_{j \in \NN} \rot(\rot([\min_{l \in \NN} \rot(t,l)],- \argmin_{j \in \NN} \rot(t,j)) x,j)) \mylabel{eqnlineb2} \\
    &
    = &
    ([\min_{j \in \NN} \rot(\rot( \min_{l \in \NN} \rot(t,l),- \argmin_{j \in \NN} \rot(t,j)) x,j)], \nonumber \\
    &
    &
    \quad - \argmin_{j \in \NN} \rot(\rot(\min_{l \in \NN} \rot(t,l),- \argmin_{j \in \NN} \rot(t,j)) x,j)) \mylabel{eqnlineb3} \\
    &
    = &
    ([\min_{j \in \NN} \rot(tx,j)],- \argmin_{j \in \NN} \rot(tx,j)) \mylabel{eqnlineb4} \\
    &
    = &
    ([\eta_S(tx)],- \argmin_{j \in \NN} \rot(tx,j)) \mylabel{eqnlineb5}
   \end{eqnarray}
   The first line in this deduction is obtained by applying the induction hypothesis to the definition of
   $\delta$.
   From line (\ref{eqnlineb1}) to line (\ref{eqnlineb2}), we applied the definition of
   $\eta_S$.
   From line (\ref{eqnlineb2}) to line (\ref{eqnlineb3}), we used the property that by the fact that we are concerned with an equivalence relation over subtrees, we know that for all
   $t \in
   T$ and
   $i \in
   \mathcal{I}$, we have
   $[[t]x]=[tx]$.
   This fact also holds for all rotations of $[t]x$ and $tx$.
   From line (\ref{eqnlineb3}) to line (\ref{eqnlineb4}), %
we simplify
   $\rot( \min_{l \in \NN} \rot(t,l),- \argmin_{j \in \NN}
   \rot(t,j))$ to
   $t$, as in this equation, the two rotations even up.
   The step to the last line just uses the definition of the $\eta_S$ function.
  \end{proof}

\subsection{A Complete Proof of EXPTIME-hardness (in the Number of Processes) for Rotation-symmetric Synthesis}

We start by concretizing the definition of space-bounded alternating Turing machines.
  \begin{definition}[\cite{DBLP:journals/jacm/ChandraKS81}]
   We say that the Turing machine 
   $M=(Q,\Sigma,\Gamma,\delta,q_0,g)$
   is $f(k)$-SPACE bounded for some function
   $f : \NN \rightarrow
   \NN$ if for every accepted word of length
   $k$,
   $M$ also accepts the word when considering all runs whose space usage exceeds
   $f(k)$ to be rejecting.
  \end{definition}
  For the scope of this paper, when discussing
  $f(k)$-space bounded Turing machines, we only consider functions
  $f(k)$ that are easy to compute (in time polynomial in
  $f(k)$).
  For every accepted word, we can arrange a set of runs of an alternating Turing machine for that word in a run tree that splits whenever a state with universal branching is visited. All runs in this tree also do not exceed the space bound.
  By definition, for every accepted word, there exists such a tree, and all of its leafs are accepting configurations.

  \begin{lemma}
   Given an
   $f(k)$-space bounded alternating Turing machine
   $M=(Q,\Sigma,\Gamma, \allowbreak \delta, \allowbreak q_0, \allowbreak g)$, we can reduce the acceptance of a word
   $w \in
   \Sigma^k$ by
   $M$ to the symmetric realizability problem of
   $n =
   f(k)$ processes with a specification in LTL of size polynomial in
   $|Q| \cdot |\Gamma| \cdot
   |w|$.
  \end{lemma}
  \begin{proof}
   The proof is done by constructing a specification
   $\psi$ that is realizable in the symmetric setting if and only if
   $w = w_0 \ldots
   w_{k-1}$ is \textbf{not} accepted.
   We define
   $\psi = \psi_1 \wedge
   \psi_2$ for
   $\psi_1$ and
   $\psi_2$ to be given below.
   We have $\AP^I_L = \{C\}$
   and 
   $\AP^O = \{S^0, \ldots, S^{|\AP^O|}
    \}$ of sufficient cardinality to encode every element from the set
   $\Gamma' = (\Gamma \cup \{ \epsilon \} \cup (Q \times \Gamma)) \times
   2^{\{\lfloor,\rfloor\}}$.
   This way, the output signals of the individual processes represent locations on the Turing machine tape.
   For simplicity, in the following, for all
   $0 \leq j <
   n$, we denote the set of output atomic propositions for process
   $j$ by
   $S_j$ (with the additional shorthand $S_{-1} := S_{k-2}$).
   We also encode the current state of the machine and tape end markers onto the tape.
   Note that in the symmetric setting we do not have a designated process for the initial state.
   The formula
   $\psi_1$ makes sure that precisely the processes retrieving a
   $C$ in the first round start a Turing computation (provided that the initial tapes would not collide), so
   $\psi_1 = (\neg C_{-k+1} \wedge \ldots \wedge \neg C_{-1} \wedge C_0 \wedge \neg C_1 \wedge \ldots \wedge \neg C_{k-1}) \Rightarrow (S_0 = (q_0,w_0,\{\lfloor\}) \wedge\ S_1 = (w_1,\emptyset) \wedge\ \ldots\ \wedge\ S_{k-2} = (w_{k-2},\emptyset)\ \wedge\ S_{k-1} =
   (w_{k-1},\{\rfloor\}))$.

   In order to deal with multiple computations starting in the first round, the delimiters
   $\lfloor$ and
   $\rfloor$ mark the ends of the tape of a machine.
   The specification part
   $\psi_2$ makes sure that the simulated Turing machine(s) do(es) not accept the input word, i.\,e.\ the processes have to simulate the computation on the Turing tape(s) but never reach an accepting state.
   For this, we let the choice of the next state for existential branching be made by the external input whereas for universal branching, the choice is made by the system.

   We syntactically extend LTL slightly for improved readability by allowing Boolean operations over sets of symbols.
   For example,
   $S_0 =
   X(S_0)$ can be unfolded to
   $\bigwedge_{p \in S_0}(p \leftrightarrow
   \LTLX s)$.
   Furthermore,
   $S_{-1},S_0,S_{+1} \xrightarrow[a]{\delta}
   \LTLX(S_{-1}),\LTLX(S_0),\LTLX(S_{+1})$ for
   $a \subseteq
   \{1,2\}$ is $\TRUE$ if for the part of the tape represented by
   $S_{-1}|_\Gamma,S|_\Gamma,S_{+1}|_\Gamma$, the transition from
   $S_0|_Q$ can be made such that afterwards
   $S_{-1},S_0,S_{+1}$ is a valid part of the configuration (at the same place on the tape), the next state is included at the respective position on the tape, and we have taken the
   $b$th of the two possible transitions for
   $b \in
   a$.
   Furthermore the tape is extended correctly such that the computation is never left of the
   $\lfloor$ marker or right of the
   $\rfloor$ marker (as seen from the initial configuration).
   If two markers collide, the next configuration is obtained by simply removing the state information from the part of the tape.
   If a rejecting state is reached, this rule is also applied.
   Note that without loss of generality, we can assume that precisely two transitions are possible in each state.
   For
   $\tilde \Gamma = (\epsilon \cup \Gamma) \times
   2^{\{\lfloor,\rfloor\}}$, we set:
   \allowdisplaybreaks
   \begin{align*}
    \psi_2 &
    = \LTLG (S_0 \text{ is in }Q \times \Gamma \times 2^{\{\lfloor,\rfloor\}} \wedge g(S_0|_Q)=\text{\quotedwedge} \rightarrow \\
    &
    \quad ( S_{-1},S_0,S_{+1} \xrightarrow[\{1,2\}]{\delta} \LTLX(S_{-1}),\LTLX(S_0),\LTLX(S_{+1}))) \\
    &
    \wedge \LTLG (S_0 \text{ is in }Q \times \Gamma \times 2^{\{\lfloor,\rfloor\}} \wedge g(S_0|_Q)=\text{\quotedvee} \wedge \neg C \rightarrow \\
    &
    \quad ( S_{-1},S_0,S_{+1} \xrightarrow[\{1\}]{\delta} \LTLX(S_{-1}),\LTLX(S_0),\LTLX(S_{+1}))) \\
    &
    \wedge \LTLG (S_0 \text{ is in }Q \times \Gamma \times 2^{\{\lfloor,\rfloor\}} \wedge g(S_0|_Q)=\text{\quotedvee} \wedge C \rightarrow \\
    &
    \quad ( S_{-1},S_0,S_{+1} \xrightarrow[\{2\}]{\delta} \LTLX(S_{-1}),\LTLX(S_0),\LTLX(S_{+1}))) \\
    &
    \wedge \LTLG ((S_0 \text{ is in } \tilde \Gamma) \wedge (S_{-1} \text{ is in } \tilde \Gamma) \wedge (S_{+1} \text{ is in }\tilde \Gamma) \rightarrow \\
    &
    \quad S_0 = \LTLX(S_0))\\
    &
    \wedge \LTLG (\neg g(S_0|_Q) \text{ is accepting})
   \end{align*}

   Assume that there exists an accepting tree of
   $M$'s computations for
   $w$.
   In this case, the environment could set the input to the first process to
   $\TRUE$ in the first round and to
   $\FALSE$ for the other processes.
   Hence, there is only one computation of the Turing machine being simulated.
   By the environment choosing the existential branching according to the acceptance tree it can be assured that eventually an accepting state is reached, which is not allowed by
   $\psi$.
   Therefore
   $\psi$ is unrealizable in this setting.

   On the other hand, if
   $w$ is not accepted by
   $M$, then there exists no tree of accepting runs of
   $M$ (that do not use more than
   $f(k)$ space) for
   $w$.
   Since in this case the implementation can decide which transition to take in case of universal branching, it can assure that the run simulated by an implementation of
   $\psi$ either ends in a rejecting state or exceeds the space limit.
   Since in case of collisions of end markers on the tape, the state information can be removed from the output and the processes can then output their last tape contents forever, the specification is trivially fulfilled in this case.
   The same applies for the case of reaching a rejecting state, so
   $\psi$ is realizable by a symmetric system.
  \end{proof}

  \subsection{Proof of Lemma 13}

\begin{proof}
  Let
   $\psi$ be the specification from Lemma~\ref{lemma:pnueliroesner}.
   We show that there exists an implementation for the specification
   $\psi$ in the A0 architecture if and only if there exists a joint implementation for the two processes in the S2 architecture that satisfies
   $\psi' = \psi_d \wedge \LTLG ( v \leftrightarrow \LTLX o ) \wedge \LTLG ( w \leftrightarrow \LTLX p
   )$, where
   $\psi_d$ results from prefixing all occurrences of the signals
   $y$ and
   $z$ in
   $\psi$ with a next-time operator. Without loss of generality, we assume that $\psi$ encodes the termination of a Turing machine. %

   $\Rightarrow$: If we have an implementation for the A0 architecture satisfying
   $\psi$, then the implementation needs to be a symmetric one: both processes output the same bitstream when reading a
   $\TRUE$ value from their respective local input signal for the first time.
   We can take an implementation for one of these processes, and turn it into an implementation for a process in the S2 architecture: just simulate the process over the input signal
   $a$ and use
   $g$ as the local output for the tape content.
   At the same time, copy all values from
   $b$ to
   $e$, and
   from $c$ to
   $f$.
   This makes sure that
   $\LTLG ( v \leftrightarrow \LTLX o ) \wedge \LTLG ( w \leftrightarrow \LTLX p
   )$ is satisfied by the resulting system.
   Since the bottom process in the S2 architecture then outputs
   $q$ with one computation cycle of delay to
   $y$,
   $y$ in the S2 architecture always represents the output of the left process in the A0 architecture with a delay of one cycle.
   For the signal line from
   $v$ to
   $z$, the same line of reasoning holds: the data from
   $v$ appears at the signal
   $o$ with a delay of one, and then, since every process in the S2 architecture simulates a process in the A0 architecture reading from
   $a$ and writing to
   $g$,
   $z$ is the output of the A0 process for the input
   $v$ with a delay of one cycle.
   Thus,
   $\psi_d$ is fulfilled by the system.
   Taking all of these facts together, the architecture with the implementation also fulfills
   $\psi$.

   $\Leftarrow$: Assume that
   $\psi'$ is realizable for the S2 architecture and we have a process that ensures that
   $\psi'$ is satisfied in the Moore machine that represents the behavior of the complete architecture with the process implementation.
   We argue that the process has to behave like a process for the A0 architecture for the input
   $a$ and output
   $g$.
   The important feature of this architecture is that the process does not know if the local input
   $b$ is the (delayed)
   $a$ input to the other process, or if its
   $c$ input is the (Turing machine tape) output of the other process. Thus, it cannot find out if it is the top process or the bottom process in the architecture and must prevent violating the specification in either case.

   Recall that without loss of generality, we assumed $\psi_D$ to encode the computation/acceptance of a Turing machine (in the same way as Pnueli and Rosner defined it~\cite{DBLP:conf/focs/PnueliR90}).
   First of all, in order to satisfy the specification, both processes have to output the first two Turing tape computations %
   on their
   $g$ outputs when obtaining a
   $\TRUE$ value to the
   $a$ input.
   This follows from the fact that the specification only allows the processes to forward information from the inputs
   $b$ and
   $c$, so the overall requirement to start with the first two tape contents on
   $y$ when reading a true-bit on
   $u$, and to start with the first two tape contents on
   $z$ when reading from
   $v$ can only be fulfilled by distributing the writing of the two tapes among the processes.

Now assume that a process receives the first Turing tape configurations on local input $c$, then a start signal on input $a$ while the second Turing tape content starts, and after both Turing tape contents have been seen, we get a starting signal on $b$. Let this input stream be called the \newterm{reference stream}.

Since the process does not know whether it is the top-most one in the architecture, it also has to output the third Turing tape configuration on its local output $g$ after the first two of these, as the input to $b$ might have been forwarded to the bottom-most process, where it triggered the other process to output the first two tape configurations. Then, $\psi_D$ would be violated if the top-most process did not output the first three configurations. But then, if the process is the lower-most one, as the local input $c$ might actually be the output of the top-most process, must also output the first three Turing tape configurations in order not to violate the specification. Note that the lower-most process must do that regardless of its local input $b$, as it is just forwarded garbage from global input $w$ then. 

But then, this means that the top-most process also has to output the first four Turing tape configurations when reading the reference stream by the same reasoning - it might be the top-most process, and since it has forwarded a signal that would trigger the other process to start the Turing tape computation one tape content later, it would otherwise violate the specification.

We can iterate this argument ad infinitum. Now if the process ensures that the overall system, when the process is instantiated in the S2 architecture, satisfies the specification, then this means that the Turing machine has to terminate --- since we can force the system to output the correct Turing tape computations along a run of the Turing machine, and we also require it to eventually halt, there is no other way that the specification can be satisfied.
\end{proof}

\subsection{Proof of Theorem 14}

We first provide a pair of compression functions for LTL.

  \begin{definition}
  \label{def:ltlCompressionFunction}
   Let
   $\AP$ be some set of signals, which, w.l.o.g, we assume to be
   $\{x_1, \ldots,
    x_n\}$.
   We define a compression function
   $f^\mathrm{LTL}$ over
   $\AP$ as follows: for every
   $w = w_0 w_1 w_2 \ldots \in
   (2^\AP)^\omega$, we set
   $f^\mathrm{LTL}(w) = w'_0 w'_1 w'_2 w'_3
   \ldots$ such that:
   \begin{itemize}
   \item
    $\forall i \in
    \NN$,
    $w'_{2i(n+1)} = w'_{2i(n+1)+1} =
    \{\chi\}$,
   \item
    $\forall i \in
    \NN$ and
    $j \in \{1, \ldots,
     n\}$, we have
    $w'_{2i(n+1)+2j} =
    \emptyset$, and
   \item
    $\forall i \in
    \NN$ and
    $j \in \{1, \ldots,
     n\}$, we have
    $\chi \in w'_{2i(n+1)+2j+1}$ if and only if
    $x_j \in
    w_i$.
   \end{itemize}
  \end{definition}
  The compression function
  $f^{LTL}$ is exemplified in Figure \ref{fig:ltlWordCompression}.
 
  Our interest in
  $f^\mathrm{LTL}$ lies in the fact
  $f^\mathrm{LTL}$ has an adjunct specification compression function
  $f'^{\mathrm{LTL}}$ for LTL properties.
  For the following lemma, we use the fact that the LTL operators $\LTLF$ and $\LTLG$ can be encoded using the $\LTLX$ and $\LTLU$ operators, so we only need to consider the latter two here. Given a word $w = w_0 w_1 w_2 w_3 \ldots \in (2^\AP)^\omega$, an index $i \in \NN$, and an LTL formula $\varphi$, we write $w,i \models \varphi$ if and only if $w_i w_{i+1} w_{i+2} \ldots \models \varphi$.

  \begin{lemma}
   \label{lem:specificationCompression} A specification compression function
   $f'^{\mathrm{LTL}}$ that corresponds to
   $f^\mathrm{LTL}$ can be defined inductively over the structure of an LTL formula as follows (for
   $\AP = \{x_1, \ldots,
    x_n\}$):
   \begin{itemize}
   \item For $\psi = \FALSE$, we set $f'^{\mathrm{LTL}}(\psi) = \FALSE$. Likewise, for $\psi = \TRUE$, we set $f'^{\mathrm{LTL}}(\psi) = \TRUE$.
   \item
    For
    $\psi = \psi_1 \wedge
    \psi_2$,
    $\psi = \psi_1 \vee
    \psi_2$, and
    $\psi = \neg
    \psi_1$ for some LTL formulas
    $\psi_1$ and
    $\psi_2$, we have
    $f'^{\mathrm{LTL}}(\psi) = f'^{\mathrm{LTL}}(\psi_1) \wedge
    f'^{\mathrm{LTL}}(\psi_2)$,
    $f'^{\mathrm{LTL}}(\psi) = f'^{\mathrm{LTL}}(\psi_1) \vee
    f'^{\mathrm{LTL}}(\psi_2)$, and
    $f'^{\mathrm{LTL}}(\psi) = \neg
    f'^{\mathrm{LTL}}(\psi_1)$, respectively.
   \item
    For
    $\psi =
    x_j$ for some
    $x_j \in
    \AP$, we set
    $f'^{\mathrm{LTL}}(\psi) = \LTLX^{2j+1}
    \chi$.
   \item For $\psi = \LTLX\, \psi_1$ for some LTL formula $\psi_1$, we set $f'^{\mathrm{LTL}}(\psi) = \LTLX^{2(n+1)} f'^{\mathrm{LTL}}(\psi_1)$.
   \item
    For
    $\psi = \psi_1\ \LTLU\
    \psi_2$, we set
    $f'^{\mathrm{LTL}}(\psi) = ((\chi \wedge \LTLX \chi \wedge \LTLX^2 \neg \chi) \rightarrow \psi_1)\ \LTLU\ ((\chi \wedge \LTLX \chi \wedge \LTLX^2 \neg \chi) \wedge
    \psi_2)$.
   \end{itemize}
  \end{lemma}
  \begin{proof}

   We show the lemma by structural induction.
   More specifically, we show that for every
   $w \in
   (2^\AP)^\omega$ and
   $i \in
   \NN$, we have
   $w,i \models
   \psi$ if and only if
   $w', 2(n+1) \cdot i \models
   f'^{\mathrm{LTL}}(\psi)$ for
   $w' =
   f^{\mathrm{LTL}}(w)$.
   \begin{itemize}
   \item
    Case
    $\psi =
    \TRUE$,
    $\psi =
    \FALSE$: trivial
   \item
    Case
    $\psi =
    x_j$ for some
    $x_j \in
    \AP$.
    The definition of
    $f^{\mathrm{LTL}}$ ensures that for all
    $j \in \{1, \ldots,
     n\}$, we have
    $x_j \in
    w_i$ if and only
    $\chi \in
    w'_{2i(n+1)+2j+1}$.
    Thus, we have
    $w',2i(n+1) \models \LTLX^{2j+1}
    \chi$ if and only if
    $w,i \models
    x_j$.
   \item
    Case
    $\psi = \psi_1 \wedge
    \psi_2$,
    $\psi = \psi_1 \vee
    \psi_2$, and
    $\psi = \neg
    \psi_1$: trivial
   \item Case $\psi = \LTLX\, \psi_1$: follows from the inductive hypothesis and the definitions of $f^\mathrm{LTL}$ and $f'^{\mathrm{LTL}}$
   \item
    Case
    $\psi = \psi_1\ \LTLU\
    \psi_2$.
    Here,
    $\chi \wedge \LTLX \chi \wedge \LTLX^2 \neg
    \chi$ is true precisely at time points
    $2(n+1)i$ for some
    $i \in
    \NN$, thus anything of interest for evaluating $f'^{\mathrm{LTL}}(\psi)$ happens at these time instants.
    For $f'^{\mathrm{LTL}}(\psi)$ to be true, we must have that for some
    $i$,
    $w'$, we get that
    $w',2(n+1)i \models
    f'^{\mathrm{LTL}}(\psi_2)$, which by the induction hypothesis is equivalent to
    $w, i \models
    \psi_2$.
    For all
    $j <
    i$, we must have
    $w', 2(n+1)j \models
    f'^{\mathrm{LTL}}(\psi_1)$, which by the induction hypothesis is equivalent to
    $w,j \models
    \psi_1$.
   \end{itemize}
  \end{proof}
  Reconsider the setting from Figure \ref{fig:ltlWordCompression}.
  As an example, the specification
  $x_4\ \LTLU\
  x_3$ for the original case translates to
  $((\chi \wedge \LTLX \chi \wedge \LTLX^2 \neg \chi) \rightarrow \LTLX^9 \chi)\ \LTLU\ ((\chi \wedge \LTLX \chi \wedge \LTLX^2 \neg \chi) \wedge \LTLX^7
  \chi)$ for the compressed case.
  In both variants, the specification is not fulfilled for this example.

  Using Lemma \ref{lem:specificationCompression}, we can now prove Theorem~\ref{theorem:undecidabilityS2}:

\begin{proof}
   Since we can compress (1) the input signals
   $u$,$v$, and
   $z$ into one signal (named $x$ in the S0 architecture),
   (2) $o$,
   $p$, and
   $q$ into one signal,
   (3) $u$, 
   $v$, and
   $z$ into one signal, and (4) adapt
   $\psi'$ accordingly  
   (and all of these compressions can use the same encoding), the undecidability of the symmetric synthesis problem for the S0 architecture follows from the undecidability of the symmetric synthesis problem for the S2 architecture.
   Definition~\ref{def:ltlCompressionFunction} and Lemma~\ref{lem:specificationCompression} describe how this word compression step can be performed.
   However, we need to make sure that (1) the correct functioning of the processes in the S0 architecture is only enforced on input streams that result from compressing a word over $2^{u,v,w}$ according to Definition \ref{def:ltlCompressionFunction} and (2) the output streams of the processes need to be correct compressions if the input to the first process is a valid compressed word.
   For this, we take the adapted version of $\psi'$ and replace it by $(\psi' \wedge \phi_{correct}) \vee \LTLF \phi_{invalid1} \vee \phi_{invalid2}$, where $\phi_{invalid1}$ and $\phi_{invalid2}$ encode that a part of the input stream is found that shows that it can not have been obtained by word compression according to Definition~\ref{def:ltlCompressionFunction}. 
   Formally, we can describe these properties as:
   \begin{eqnarray*}
   \phi_{invalid1} & = & \chi \wedge \LTLX \chi \wedge \LTLX \LTLX \neg \chi \wedge \left( \neg \LTLX^8 \chi \vee \neg \LTLX^9 \chi \vee \LTLX^{10} \chi \vee \bigvee_{i \in \{1, \ldots, 3\}} \LTLX^{2i} \chi \right) \\
   \phi_{invalid2} & = & \neg \chi \vee \neg \LTLX \chi \vee \LTLX \LTLX \chi
   \end{eqnarray*}
   Intuitively, $\phi_{invalid1}$ states that starting from a character start sequence, the next character start sequence does not come after exactly 8 cycles (or we have an illegal bit encoding along the way), and $\phi_{invalid2}$ states that the input stream does not start with a character start sequence.
   In the same way, we can encode that $y$ and $z$ represent correctly compressed streams:
   \begin{multline}
   \phi_{correct} = \bigwedge_{p \in \{x,y\}} p \wedge \LTLX p \wedge \LTLX^2 \neg p \wedge \LTLG \bigg( \neg p \vee \neg \LTLX p \vee \LTLX^2 p \vee \\
   \Big( \bigwedge_{i \in \{1, \ldots, 3\}} \LTLX^{2i} p \Big) \wedge \LTLX^8 p \wedge \LTLX^9 p \wedge \LTLX^{10} \neg p \bigg)
   \end{multline}

\end{proof}

\end{document}